\documentclass[twocolumn,10pt,superscriptaddress,amsmath,amssymb,aps,prl,floatfix,longbibliography]{revtex4-2}

\usepackage{amsmath,amssymb,amsfonts,amsthm}
\usepackage{graphicx}
\usepackage{bm}
\usepackage[colorlinks=true,linkcolor=blue,citecolor=blue,urlcolor=blue]{hyperref}
\usepackage{xcolor}

\newcommand{\tr}{\mathrm{Tr}}
\newcommand{\var}{{\rm Var}}
\newcommand{\vect}[1]{\mathbf{#1}}

\newtheorem{theorem}{Theorem}
\newtheorem{lemma}{Lemma}
\newcommand{\n}{\nonumber\\}
\newcommand{\ket}[1]{|#1\rangle}
\newcommand{\bra}[1]{\langle#1|}

\newcommand{\wex}[2]{\langle #1, #2 \rangle_w}

\newcommand{\BC}{\mathcal B_{\rm C}}
\newcommand{\BQ}{\mathcal B_{\rm Q}}
\newcommand{\FC}{\mathcal F_{\rm C}}

\newcommand{\BCd}{\mathcal B_{\rm C}^{\rm GCR}}
\newcommand{\BCc}{\mathcal B_{\rm C}^{\rm GBar}}
\newcommand{\BCg}{\mathcal B_{\rm C}^{\rm FG}}

\newcommand{\BQd}{\mathcal B_{\rm Q}^{\rm GCR}}
\newcommand{\BQc}{\mathcal B_{\rm Q}^{\rm GBar}}
\newcommand{\BQg}{\mathcal B_{\rm Q}^{\rm FG}}
\begin{document}

\title{Global Bounds beyond Local Quantum Metrology}

\author{Hai-Long Shi}
\email{hailong.shi@ino.cnr.it}
\affiliation{INO-CNR and LENS, Largo Enrico Fermi 2, 50125 Firenze, Italy}

\author{Augusto Smerzi}
\email{augusto.smerzi@ino.cnr.it}
\affiliation{College of Engineering Physics, Shenzhen Technology University, Shenzhen 518118, China}

\begin{abstract}
Quantum Cramér--Rao theory is intrinsically local: it bounds precision near a specified parameter value, and its saturating measurement generally depends on that value. Barankin-type bounds use finite parameter displacements, but remain anchored to a chosen reference value. This leaves open a basic global-estimation problem: when the parameter is known only within a broad domain, what precision can be guaranteed by a single estimator and a single measurement strategy fixed before the true value is localized? We answer this question by introducing global score functions tied to a weighted variance over the whole parameter domain. 
Their correlations generate a hierarchy of precision bounds: global Cramér--Rao and Barankin-type bounds arise as restricted levels, whereas unrestricted score correlations yield a fully global bound for the prescribed weighted variance. The hierarchy recovers local Cramér--Rao theory in the many-repetition limit and reveals genuinely global precision limits for finite data over broad domains. In the quantum setting, the construction identifies when this fully global bound can be realized by a single parameter-independent measurement. 
The same framework extends to Bayesian estimation, recovering the Van Trees bound in the local limit while yielding stronger finite-width lower bounds on the Bayesian mean-square error beyond this limit.
\end{abstract}

\maketitle

\textit{Introduction}.--
Quantum metrology turns quantum resources, such as entanglement and squeezing, into enhanced sensitivity~\cite{RevModPhys.90.035005,toth2014quantum,paris2009quantum,huang2024entanglement,PhysRevLett.96.010401,PhysRevLett.102.100401,montenegro2025quantum,RevModPhys.89.035002,giovannetti2011advances}.
Its standard benchmark is the quantum Cram\'{e}r--Rao bound (QCRB), relating the variance of an unbiased estimator to the inverse quantum Fisher information~\cite{PhysRevLett.72.3439,helstrom1969quantum,holevo2011probabilistic,cramer1999mathematical,rao1945information,fisher1920mathematical}.
This theory is intrinsically local: it bounds estimation around a specified parameter value, and the saturating measurement generally depends on that value~\cite{PhysRevLett.72.3439,PhysRevA.111.022436,zhou2020saturating,liu2025optimal,wiseman2009quantum,PhysRevLett.75.2944,demkowicz2015quantum,PhysRevA.54.4564}.
Thus, the local Cram\'er--Rao bound is operationally sharp only after the unknown parameter has been confined to a sufficiently small region.

Many metrological tasks begin outside this regime.
In phase estimation, frequency calibration, thermometry, and related sensing problems, the parameter may initially be known only over a broad domain~\cite{PhysRevLett.85.5098,PhysRevLett.124.030501,pang2017optimal,PhysRevResearch.6.043171,PhysRevA.95.012136,PhysRevLett.127.190402,PhysRevLett.126.200501,PhysRevA.108.012613}.
Experiments also often operate with finite data, where asymptotic local theory does not determine actual estimator performance~\cite{PhysRevA.101.032114,qbn1-p6bq,PhysRevResearch.6.043261,PhysRevLett.134.010804}.
A large quantum Fisher information (QFI) at one point neither guarantees distinguishability over the full domain nor identifies a measurement useful before the true parameter is known~\cite{PhysRevResearch.6.L032048,hayashi2025cramer,PhysRevLett.108.230401}.
The relevant object is therefore not only the local statistical speed between neighboring states, but the global distinguishability structure of the full family~\cite{mukhopadhyay2025current}.

This exposes a structural limitation of standard frequentist metrology.
Local precision theory is anchored to the unknown true value: the variance, Fisher information, and locally optimal measurement are all defined only after choosing a reference point.
Barankin-type bounds use finite parameter displacements and therefore probe nonlocal distinguishability~\cite{barankin1949locally,PhysRevLett.130.260801,bhattacharyya1946some,chapman1951minimum}, but their optimized likelihood-response directions remain defined with respect to that point~\cite{barankin1949locally}.
In the quantum setting, this reference dependence also affects the measurement: as for the QCRB~\cite{PhysRevLett.72.3439}, a measurement saturating a quantum Barankin-type bound generally depends on the true parameter~\cite{PhysRevLett.130.260801}.
Thus, finite-displacement bounds enlarge the score directions, but do not by themselves solve the global measurement-design problem.

In global metrology, the variance criterion, estimator, and measurement must be specified before the unknown parameter is localized to a neighborhood of any given value.
Our construction implements this requirement directly: it is not a weighted average of Barankin bounds evaluated at different reference values.
Rather, it defines a genuinely global finite-data task, with the figure of merit, estimator, and measurement fixed over the full parameter domain as a single operational strategy.

We develop such a framework by replacing the true-value-dependent score with a global score field over the parameter domain, tied to the global figure of merit \( \sum_i w_i\var_{\theta_i}(\hat\theta) \).
Couplings between score components at different parameter values generate a hierarchy of precision bounds.
Restricted couplings recover weighted global Cramér--Rao and Barankin-type bounds, while unrestricted couplings yield a fully global bound controlled by correlations across the full domain.
The hierarchy reduces to local Cramér--Rao theory in the many-repetition limit, but captures irreducibly global estimator fluctuations in the finite-data, broad-domain regime.

Optimizing over measurements gives the quantum counterpart.
While quantum Cramér--Rao and quantum Barankin saturation generally require measurements optimized at the true parameter value, the fully global problem asks whether one measurement, fixed before the true value is known, can realize the prescribed variance bound over the whole domain.
We identify a broad compatible class where this is possible: for the common-bias fully global problem, the saturating measurement is fixed by full-domain score correlations rather than by the local quantum Fisher information at one point.
The same score principle extends to Bayesian estimation, recovering Van Trees as a local limit and yielding tighter finite-width bounds.


\textit{Classical global-score hierarchy}.---
We derive the classical hierarchy for a frequentist figure of merit over the full parameter domain.
Let \( \theta\in\Theta=\{\theta_1,\ldots,\theta_n\}\subseteq\mathbb R \), with the continuous case obtained by refinement, and let \(p(x|\theta)\) be the outcome distribution.
A single estimator \( \hat\theta(x) \) is used for all parameter values, with global performance
\begin{equation}\label{eq:weighted_var}
\overline{\var}(\hat\theta)
=
\sum_i w_i\,\var_{\theta_i}[\hat\theta],
\end{equation}
where \(w_i\ge0\), \( \sum_i w_i=1 \), and
\(\var_{\theta_i}[\hat\theta]\) denotes the variance of the estimator with respect to \(p(x|\theta_i)\).
The weights define a frequentist risk, not a Bayesian prior.

Let \( \mathbf g(x)=(g_1(x),\ldots,g_n(x))^\top \) collect test functions, and let \( \mathbf a_i=(A_{1i},\ldots,A_{ni})^\top \) be the \(i\)-th column of a real hierarchy matrix \(A\).
At the parameter value \( \theta_i \), we test the estimator fluctuation with
\begin{equation}\label{eq:score_func}
S_i(x)
=
\frac{\mathbf a_i^\top\mathbf g(x)}{p(x|\theta_i)} .
\end{equation}
The matrix \(A\) couples the score directions associated with different parameter values. Diagonal choices recover reference-value score tests, while off-diagonal entries generate genuinely global score directions.

Applying the Cauchy--Schwarz inequality to
\( \Delta_i\hat\theta=\hat\theta-\mathbb E_{\theta_i}(\hat\theta) \)
and \(S_i\), with weights \(w_i\), we obtain the general global-score bound
\begin{equation}\label{main_1}
\overline{\var}(\hat\theta)
\ge
\frac{
\left(\sum_i w_i\,\mathbf a_i^\top\mathbf b^{(i)}\right)^2
}{
\sum_i w_i\,\mathbf a_i^\top\mathcal C^{(i)}\mathbf a_i
}
\equiv
\mathcal B_{\rm C}(A),
\end{equation}
where
\begin{equation}
\mathcal C^{(i)}
\!=\!
\sum_x
\frac{\mathbf g(x)\mathbf g^\top(x)}{p(x|\theta_i)},\,     \mathbf b^{(i)}
\!=\!
\sum_x
\mathbf g(x)
\left[
\hat\theta(x)-\mathbb E_{\theta_i}(\hat\theta)
\right].\nonumber
\end{equation}
Here \( \mathcal C^{(i)} \) is the information matrix of the chosen score directions at \( \theta_i \), and \( \mathbf b^{(i)} \) gives their overlaps with the estimator fluctuation. Different restrictions on the hierarchy matrix \(A\) select different levels of the global hierarchy.

\textit{Global Cramér--Rao bound}.---
The first restricted level restricts \(A\) to be diagonal,
\(A=\operatorname{diag}(a_1,\ldots,a_n)\), so that each score \(S_i\) uses only the test function associated with \(\theta_i\).
The bound~\eqref{main_1} then becomes a Rayleigh quotient in the coefficients \(a_i\), whose supremum gives
\begin{equation}\label{eq:diagonal_bound}
\mathcal B_{\rm C}^{\rm GCR}
=
\sup_{\mathbf a\in\mathbb R^n}
\mathcal B_{\rm C}[\operatorname{diag}(a_1,\ldots,a_n)]
=
\sum_{i=1}^n w_i
\frac{\big(b_i^{(i)}\big)^2}{[\mathcal C^{(i)}]_{ii}} .
\end{equation}
Hence, \( \overline{\var}(\hat\theta)\ge \mathcal B_{\rm C}^{\rm GCR} \), where \( b_i^{(i)} \) is the \(i\)-th component of \( \mathbf b^{(i)} \).
For the local score choice \( g_i(x)=\partial_\theta p(x|\theta_i) \), one has
\(b_i^{(i)}=\partial_\theta \mathbb E_{\theta_i}(\hat\theta)\) and
\([\mathcal C^{(i)}]_{ii}=\mathcal F_{\rm C}(\theta_i)\), with \( \mathcal F_{\rm C} \) the classical Fisher information.
Thus,
\begin{equation}\label{eq:global_cr_bound}
\overline{\var}(\hat\theta)
\ge
\sum_{i=1}^n
w_i
\frac{
\big[\partial_\theta \mathbb E_{\theta_i}(\hat\theta)\big]^2
}{
\mathcal F_{\rm C}(\theta_i)
}.
\end{equation}
This is the global Cramér--Rao level, namely a weighted global version of the biased Cramér--Rao bound, global in the risk through the sum over \( \theta_i \), but local in the score structure because each term uses only the tangent direction at the same parameter value.

\textit{Global Barankin bound}.---
The second restricted level is obtained by taking identical columns,
\( A=(\mathbf a,\mathbf a,\ldots,\mathbf a) \), with \( \mathbf a\in\mathbb R^n \), so that the same nonlocal score combination is tested at every parameter value.
The bound~\eqref{main_1} becomes a Rayleigh quotient in \( \mathbf a \), whose supremum gives
\begin{equation}\label{eq:column_bound}
\mathcal B_{\rm C}^{\rm GBar}
=
\sup_{\mathbf a\in\mathbb R^n}
\mathcal B_{\rm C}[(\mathbf a,\ldots,\mathbf a)]
=
\mathbf b_w^\top \mathcal C_w^+ \mathbf b_w ,
\end{equation}
where
\( \mathbf b_w=\sum_i w_i\,\mathbf b^{(i)} \),
\( \mathcal C_w=\sum_i w_i\,\mathcal C^{(i)} \), and
\( \mathcal C_w^+ \) is the Moore--Penrose pseudoinverse.
Hence \( \overline{\var}(\hat\theta)\ge \mathcal B_{\rm C}^{\rm GBar} \).

This global Barankin level is nonlocal in the score directions, since \( \mathbf a \) combines test functions attached to different parameter values.
For \(w_1=1\), Eq.~\eqref{eq:column_bound} reduces to the standard Barankin bound at the reference value $\theta_1$~\cite{barankin1949locally}.
For general weights, it is the Barankin-type global bound associated with the prescribed weighted variance, but remains restricted because the same score combination is tested at every parameter value.

\textit{Fully global bound}.---
The unrestricted level allows arbitrary hierarchy matrices \(A\). The bound is then a Rayleigh quotient in all columns \( \{\mathbf a_i\} \), whose supremum gives
\begin{equation}\label{eq:general_bound}
\mathcal B_{\rm C}^{\rm FG}
=
\sup_A \mathcal B_{\rm C}(A)
=
\sum_i w_i\,\mathbf b^{(i)\top}
\big(\mathcal C^{(i)}\big)^+
\mathbf b^{(i)} .
\end{equation}
Hence \( \overline{\var}(\hat\theta)\ge \mathcal B_{\rm C}^{\rm FG} \), where \( (\mathcal C^{(i)})^+ \) is the Moore--Penrose pseudoinverse. This fully global level imposes no structural restriction on how score directions associated with different parameter values are combined.
Thus, the three levels obey
\begin{equation}\label{eq:hierarchy}
\overline{\var}(\hat\theta)
\ge
\mathcal B_{\rm C}^{\rm FG}
\ge
\mathcal B_{\rm C}^{\rm GBar},
\quad
\overline{\var}(\hat\theta)
\ge
\mathcal B_{\rm C}^{\rm FG}
\ge
\mathcal B_{\rm C}^{\rm GCR}.
\end{equation}
These inequalities follow because the diagonal and column-constant choices of \(A\) are restricted subsets of the unrestricted class.

The fully global level also admits a constructive saturation condition. For globally unbiased estimators, one may choose an anchor point and impose the parameterwise Cauchy--Schwarz equality condition directly in outcome space. The problem then reduces to a range condition for an anchor-dependent linear map, together with a compatibility condition fixing the corresponding global unbiasedness relation. This gives an explicit route to unbiased saturation, detailed in the Supplemental Material (SM).

We now compare two limits of the hierarchy. Consider independent data \( \mathbf x=(x_1,\ldots,x_m) \), with \( p(\mathbf x|\theta)=\prod_{\ell=1}^m p(x_\ell|\theta) \), where \(m\) is the number of repetitions and \(n=|\Theta|\) is the number of sampled parameter values.


\textit{Observation~1}.---
For fixed finite \( n \) and local scores \( g_j(x)=\partial_\theta p(x|\theta_j) \), the fully global bound reduces in the many-repetition limit to the global Cramér--Rao level,
\begin{equation}\label{eq:observation_1}
\mathcal B_{\rm C}^{\rm FG}
\overset{m\gg1}{\simeq}
\mathcal B_{\rm C}^{\rm GCR}.
\end{equation}
Under the regularity assumptions stated in the SM, this bound is saturated asymptotically by the maximum-likelihood estimator. This follows from likelihood concentration. As \(m\) increases, data generated at \( \theta_i \) become exponentially distinguishable from data generated at separated parameter values, so the information matrices are dominated by the local score direction at \( \theta_i \). Hence the global score correlations distinguishing \( \mathcal B_{\rm C}^{\rm FG} \) from \( \mathcal B_{\rm C}^{\rm GCR} \) become asymptotically irrelevant, recovering the standard large-sample reduction to a local Cramér--Rao problem~\cite{PhysRevA.61.042312}.

\textit{Observation~2}.---
Consider the opposite regime, where \(m\) is fixed, including nonasymptotic values for which the data do not localize the parameter, while \( \Theta=\{\theta_1,\ldots,\theta_n\} \) gives an increasingly dense sampling of a broad domain.
The problem then cannot be reduced to a neighborhood of the true value.
Errors between separated parameter values remain relevant, and estimator fluctuations must be controlled over the whole domain.
As \(n\) increases, the global score family may span the fluctuations required for saturation.
When the corresponding compatibility, or range, condition is satisfied, the fully global level becomes attainable,
\begin{equation}\label{eq:observation_2}
\overline{\var}(\hat\theta)
\simeq
\mathcal B_{\rm C}^{\rm FG}.
\end{equation}
Thus, in finite-data broad-domain estimation, refining the sampled domain can drive the attainable variance toward the fully global bound, provided the score-space compatibility condition derived in the SM holds.

\begin{figure}[t]
\centering
\includegraphics[width=0.9\linewidth]{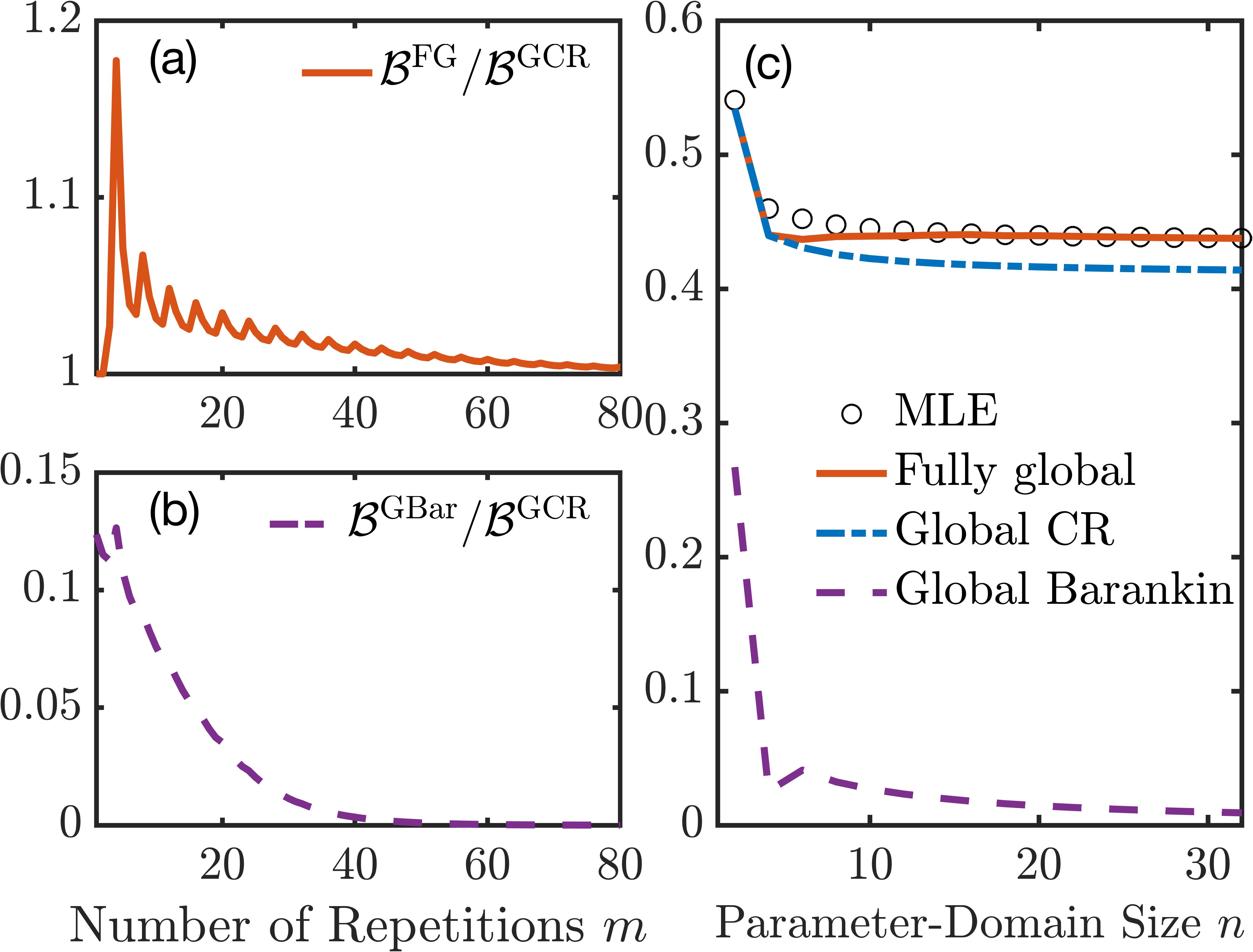}
\caption{
Classical global-score hierarchy for a noisy binary model
\( p(\mathbf{x}|\theta)=\prod_{\ell=1}^{m}(1+x_\ell d\sin\theta)/2 \),
with \( d=1/2 \), \( w_i=1/n \), and \( \theta_i\in[0,\pi] \).
(a) For fixed \( n=8 \), \( \mathcal B_{\rm C}^{\rm FG}/\mathcal B_{\rm C}^{\rm GCR}\to1 \) as \(m\) increases, recovering local Cram\'er--Rao behavior.
(b) The global Barankin level remains below the Cram\'er--Rao scale.
(c) For fixed \( m=10 \), refining the sampled domain drives the variance of the maximum-likelihood estimator (MLE) toward \( \mathcal B_{\rm C}^{\rm FG} \), whereas the global Cram\'er--Rao and Barankin levels do not track the attainable variance.
}
\label{fig:classical_bounds}
\end{figure}

\textit{Example: two-outcome measurements}.---
We illustrate the two regimes with a binary model and local scores
\( g_j(x)=\partial_\theta p(x|\theta_j) \).
For \( x=\pm1 \), let
\( p(x|\theta)=[1+xf(\theta)]/2 \),
\( f_i=f(\theta_i) \),
\( f_i'=\partial_\theta f(\theta_i) \), and
\( \mathbf f'=(f_1',\ldots,f_n')^\top \).
For a single shot, we obtain $\mathcal C^{(i)}(1)
=
\frac{\mathbf f'\mathbf f'^\top}{1-f_i^2}$.
As shown in the SM, the bias vectors lie in
\( \operatorname{Range}\mathcal C^{(i)}(1) \), giving
\( B=\alpha(\mathbf f',\ldots,\mathbf f') \) for some scalar \( \alpha \).
The hierarchy reduces to
\begin{align}
&\mathcal B_{\rm C}^{\rm FG}
=
\mathcal B_{\rm C}^{\rm GCR}
=
\alpha^2\sum_i w_i(1-f_i^2)\nonumber\\
&\geq 
\mathcal B_{\rm C}^{\rm GBar}=
\alpha^2/[\sum_i w_i/(1-f_i^2)] .
\end{align}
Thus, for a single-shot binary model, the fully global and global Cramér--Rao levels coincide, while the global Barankin level is weaker.

For \( m \) independent repetitions,
\( p(\mathbf x|\theta)=\prod_{\ell=1}^m p(x_\ell|\theta) \), one finds
\begin{equation}\label{eq:CIM_two_measurement}
[\mathcal C^{(i)}(m)]_{jk}
=
m\,h(m,\Delta_{i;jk})[\mathcal C^{(i)}(1)]_{jk},
\end{equation}
where
\( h(m,x)=(1+x)^{m-2}(1+mx) \) and
\( \Delta_{i;jk}=(f_j-f_i)(f_k-f_i)/(1-f_i^2) \).
Figure~\ref{fig:classical_bounds} shows the noisy qubit model
\( f(\theta)=d\sin\theta \), with \( d=1/2 \) and uniform weights.
For fixed domain size $n$, $ \mathcal B_{\rm C}^{\rm FG}$ tends to 
$\mathcal B_{\rm C}^{\rm GCR}$ as \(m\) increases, consistent with Observation~1.
For fixed finite \(m\), increasing the number of sampled parameter values, $n$, enriches the global score family and drives the maximum-likelihood variance toward \( \mathcal B_{\rm C}^{\rm FG} \), illustrating Observation~2.
However, the global Cramér--Rao level does not track the attainable variance in this regime.
This example, therefore, shows both limiting behaviors of the hierarchy, namely local Cramér--Rao behavior after likelihood concentration and genuinely global behavior for finite data over broad domains.

\textit{Quantum extension}.--
The hierarchy becomes quantum when \(p(x|\theta)=\tr[E_x\rho(\theta)]\) for a positive operator-valued measure (POVM)  \( \{E_x\} \) and a quantum state $\rho(\theta)$.
We choose \( g_k(x)=\tr(E_xG_k) \), where \(G_k\) is the operator score direction associated with \( \theta_k \).
Optimizing over POVMs gives
\begin{align}
\mathcal B_{\rm Q}^{\rm GCR}
&=
\sum_i
w_i
\frac{\big(b_i^{(i)}\big)^2}{[\mathcal Q^{(i)}]_{ii}},
\\
\mathcal B_{\rm Q}^{\rm GBar}
&=
\mathbf b_w^\top \mathcal Q_w^+ \mathbf b_w,
\\
\mathcal B_{\rm Q}^{\rm FG}
&=
\sum_i w_i\,
\mathbf b^{(i)\top}
\big(\mathcal Q^{(i)}\big)^+
\mathbf b^{(i)} ,
\end{align}
where $[\mathcal Q^{(i)}]_{jk}
=
\tr\!\big[
G_j\,\Omega_{\rho(\theta_i)}(G_k)
\big]$ and \( \mathcal Q_w=\sum_i w_i\mathcal Q^{(i)} \).
The map \( \Omega_{\rho}(G) \) is the generalized symmetric logarithmic derivative (SLD) on the support of \( \rho \), defined by 
\begin{equation}
G=\frac{1}{2}\{\rho,\Omega_\rho(G)\}.    
\end{equation}
Technical conditions associated with non-full-rank states are discussed in the SM.

The quantum task is not only to evaluate the quantum information matrix \( \mathcal Q^{(i)} \), but to find a measurement realizing the relevant hierarchy level.
In local metrology, the optimal measurement is tied to the SLD at the true parameter value.
In the global problem, the score field gives generalized SLDs \(L_i\), one for each \( \theta_i \),
\begin{equation}
L_i
=
\Omega_{\rho(\theta_i)}(G_{A;i}),
\quad
G_{A;i}=\sum_j A_{ji}G_j .
\end{equation}
For generic \(A\), these operators need not coincide, commute, or share eigenvectors.
Thus, global quantum metrology contains a measurement-compatibility problem absent from the classical hierarchy.
The fully global question is whether one measurement, fixed before the true value is known, can saturate a variance functional over the whole domain.

At the unrestricted level, this compatibility problem has an explicit solution for a common-bias class.
Let \( \mathbf b^{(i)}=\mathbf b \) and $\mathcal Y=\operatorname{span}\{G_j\}_{j=1}^n$.
Assume that
\begin{equation}\label{eq:identity_extended_main}
\Omega_{\rho(\theta_i)}(\mathcal Y)
\subseteq
\mathcal Y\oplus \mathbb R I
\end{equation}
for all active \( \theta_i \).
This condition allows only scalar components outside the score span, which shift eigenvalues without rotating eigenvectors.
Then the optimized score directions obey
\begin{equation}
\Omega_{\rho(\theta_i)}(G_{A;i})
=
M_{\rm FG}+c_iI ,
\end{equation}
with \(c_i\in\mathbb R\) and
\begin{equation}\label{eq:obs_3}
M_{\rm FG}
=
\sum_{j=1}^n (T^+\mathbf b)_j\,G_j ,
\qquad
T_{jk}=\tr(G_jG_k).
\end{equation}
Hence all optimized generalized SLDs share the eigenprojectors of \(M_{\rm FG}\).
Writing
\begin{equation}
M_{\rm FG}=\sum_x p_xE_x^{\rm opt},
\end{equation}
the projectors \(E_x^{\rm opt}\) define a single POVM satisfying
\begin{equation}
\mathcal B_{\rm Q}^{\rm FG}
=
\mathcal B_{\rm C}^{\rm FG}(\{E_x^{\rm opt}\}) .
\end{equation}
For \(G_j=\partial_\theta\rho(\theta_j)\), this gives an explicit global measurement determined by the full family of states. The proof is given in the SM.
The fully global measurement is therefore selected by score correlations across the parameter domain, not by the local QFI at one point.

As a minimal illustration, consider
\(\rho_\theta=(1+\vect r_\theta\cdot\vect\sigma)/2\), with
\(\vect r_\theta=(d\cos\theta,d\sin\theta,0)\),
\(\Theta=\{0,\pi/4\}\), and \(\vect b=(1,1/\sqrt2)\).
The compatibility condition is satisfied, and Eq.~\eqref{eq:obs_3} selects
\(E_x^{\rm opt}=(1+x\sigma_y)/2\), \(x=\pm1\).
This parameter-independent measurement saturates
\begin{equation}
\mathcal B_{\rm Q}^{\rm FG}
=
\mathcal B_{\rm C}^{\rm FG}
=
\frac{w_1}{d^2}
+
\frac{w_2(2-d^2)}{2d^2}.
\end{equation}
For \(w_2\neq0\) and $d\neq 1$, it does not saturate the global quantum Cramér--Rao or global quantum Barankin levels.
This example separates three design principles. Local-QFI and reference-value Barankin measurements are tied to a selected parameter value, whereas \(E_x^{\rm opt}\) is fixed before localization and saturates the weighted-domain bound.

\begin{figure}[t]
\centering
\includegraphics[width=0.8\linewidth]{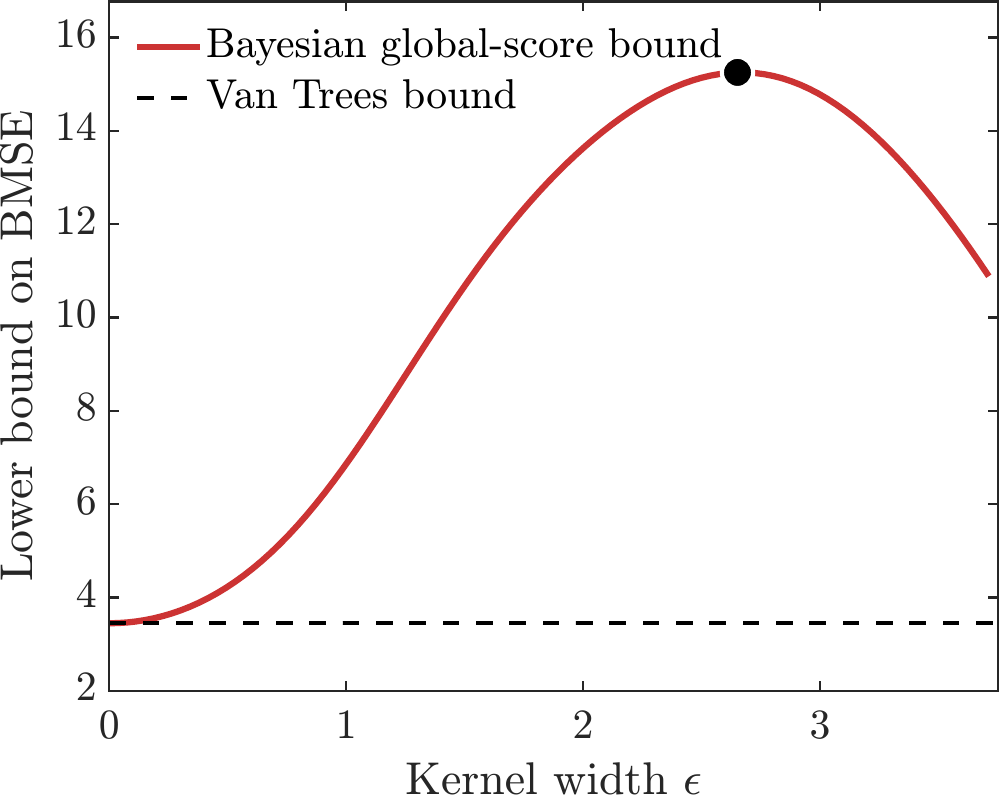}
\caption{
Bayesian global-score bound for a noisy qubit with Gaussian prior width \( \sigma_p=5 \) and decoherence \( d=0.5 \).
The local limit \( \epsilon\to0 \) recovers the Van Trees bound.
Finite-width kernels retain correlations between distinct parameter values and yield stronger BMSE lower bounds.
The black dot marks the optimal kernel width yielding the strongest BMSE lower bound.
}
\label{fig2}
\end{figure}

\textit{Bayesian extension}.---
The same score principle extends to Bayesian estimation.
For a continuous parameter with prior \( p(\theta) \), the figure of merit is the Bayesian mean-square error
\begin{equation}
\mathrm{BMSE}(\hat\theta)
=
\int d\theta\,p(\theta)\,
\mathbb E_\theta\!\left[(\hat\theta-\theta)^2\right].
\end{equation}
Although this criterion is global, the usual Van Trees bound is generated by a local score.
To retain finite correlations between distinct parameter values, we introduce a hierarchy kernel \( A_\epsilon(\theta,\theta') \) and define
\begin{equation}
S_{\rm B}(x,\theta)
=
\frac{
\int d\theta'\,
A_\epsilon(\theta,\theta')\,
\partial_{\theta'}p(x,\theta')
}
{p(x,\theta)},
\end{equation}
where \( p(x,\theta)=p(x|\theta)p(\theta) \).
Cauchy--Schwarz with respect to \( p(x,\theta) \), followed by optimization over quantum measurements, gives the Bayesian global-score bound
\begin{equation}\label{g_BY}
\mathrm{BMSE}(\hat\theta)
\ge
\frac{1}{
\iiint d\theta\,d\theta'\,d\theta''
A_\epsilon(\theta,\theta')
\mathcal Q_{\rm B}^{(\theta)}(\theta',\theta'')
A_\epsilon(\theta,\theta'')
},
\end{equation}
with
\begin{equation}
\begin{aligned}
\mathcal Q_{\rm B}^{(\theta)}(\theta',\theta'')
&=
\tr\!\left[
\partial_{\theta'}\rho[\theta']\,
\Omega_{\rho[\theta]}\!\left(\partial_{\theta''}\rho[\theta'']\right)
\right],
\\
\rho[\theta]&=p(\theta)\rho(\theta).
\end{aligned}
\end{equation}
For \( A_\epsilon(\theta,\theta')=\epsilon^{-1}K[(\theta-\theta')/\epsilon] \to \delta(\theta-\theta') \), Eq.~\eqref{g_BY} reduces to the Van Trees bound in the local limit \( \epsilon\to0 \)~\cite{van2004detection}.
At finite width, the kernel retains nonlocal score correlations and can strengthen the BMSE lower bound, as shown in Fig.~\ref{fig2}.


\textit{Conclusion}.--
We have developed a global-score framework for quantum metrology beyond the local Cramér--Rao regime.
The central object is a score field over the parameter domain, tied to a variance criterion defined before the unknown parameter is localized.
By allowing score components at different parameter values to couple through a hierarchy matrix, the construction yields a hierarchy of precision bounds for prescribed weighted variances over the domain.
The weighted global Cramér--Rao and global Barankin bounds appear as restricted levels, while the unrestricted level gives a fully global bound governed by correlations across the full parameter domain.

This hierarchy separates local and genuinely global estimation.
For many independent repetitions at fixed domain size, likelihood concentration suppresses the nonlocal score correlations and the fully global bound reduces to the Cramér--Rao structure, recovering the usual local theory.
For finite data over broad domains, by contrast, estimator fluctuations cannot be controlled by local Fisher information alone.
When the compatibility conditions identified above are satisfied, the attainable precision is governed by global score correlations between distinct parameter values.

In the quantum setting, this leads to a measurement principle beyond local-QFI optimization.
For a compatible common-bias class, the fully global bound is saturated by a single measurement fixed over the whole parameter domain and determined by the full-domain score structure, rather than by the symmetric logarithmic derivative at one parameter value.
Thus, the quantum problem is not only to bound distinguishability, but to identify when global score correlations can be made jointly measurable by a single POVM.
The same construction extends to Bayesian estimation, where the Van Trees bound is recovered as the local-kernel limit and finite-width score correlations can yield stronger BMSE lower bounds.
This pre-localization regime is the natural operating regime of finite-data quantum sensors used for broad-range phase and frequency calibration, thermometry, and sensing tasks in which adaptive narrowing has not yet been completed.

Generalized global scores therefore provide a unified route to precision limits in classical and quantum metrology beyond the local regime, in both frequentist and Bayesian settings.

\begin{acknowledgments}
H.L.S. acknowledges support from the European Commission through the H2020 QuantERA ERA-NET Cofund project ``MENTA'' and from the Horizon Europe programme HORIZON-CL4-2022-QUANTUM-02-SGA under Grant Agreement No. 101113690 (PASQuanS2.1). A.S. acknowledges support from the National Natural Science Foundation of China under Grant No. W2531008 and from the Peacock Plan.

\end{acknowledgments}

\clearpage
\onecolumngrid

\section*{Supplemental Material for ``Global Bounds beyond Local Quantum Metrology''}
\addcontentsline{toc}{section}{Supplemental Material for ``Global Bounds beyond Local Quantum Metrology''}

\setcounter{section}{0}
\setcounter{subsection}{0}
\setcounter{equation}{0}
\setcounter{figure}{0}
\setcounter{table}{0}
\renewcommand{\thesection}{S\arabic{section}}
\renewcommand{\thesubsection}{S\arabic{section}.\arabic{subsection}}
\renewcommand{\theequation}{S\arabic{equation}}
\renewcommand{\thefigure}{S\arabic{figure}}
\renewcommand{\thetable}{S\arabic{table}}


\section{Derivation of Classical Bounds}
\label{sec:SM_derivation_classical_bounds}

In this section, we derive the classical bounds presented in the main text. Our strategy is first to construct a family of bounds using generalized global score functions, and then optimize these bounds under different constraints imposed on the hierarchy matrix $A$.

\subsection{Classical Bounds for Arbitrary Test Functions}
\label{sec:SM_arbitrary_test_functions}

Consider a single parameter $\theta$ taking values in the set $\Theta=\{\theta_1,\dots,\theta_n\}$. Let $\hat\theta(x)$ be an estimator constructed from measurement outcomes $x\in X$.
Throughout this Supplemental Material, we assume the regularity condition $p(x|\theta)>0$ for all outcomes $x$ and all parameter values $\theta\in\Theta$.
To assess the performance of the estimator across all parameter values, we define its weighted variance as
\begin{equation}\label{eq:weighted_var_SM}
\overline{\var}(\hat\theta)=\sum_{i=1}^n w_i\,\var_{\theta_i}(\hat\theta),
\end{equation}
where $w_i\ge 0$ are weights that reflect the relative importance of each parameter value, and
\begin{equation}
\var_{\theta_i}(\hat\theta)
=
\sum_{x\in X} p(x|\theta_i)\big[\hat\theta(x)-\mathbb E_{\theta_i}(\hat\theta)\big]^2,
\end{equation}
is the variance of the estimator at a specific parameter value $\theta_i$ with $\mathbb E_{\theta_i}(\hat \theta)=\sum_x p(x|\theta_i)\hat\theta(x)$.
The goal is to establish fundamental lower bounds for the weighted variance~\eqref{eq:weighted_var_SM} for any estimator.

For each parameter value $\theta_j$, we choose a test function $g_j(x)$ and collect them into the vector $\vect g(x)=(g_1(x),\dots,g_n(x))^\top$.
Using these test functions, we introduce the hierarchy matrix $A$ to define a generalized global score function,
\begin{equation}\label{eq:SM_score_func}
S_i(x)=\frac{\vect a_i^\top\vect g(x)}{p(x|\theta_i)},\qquad i=1,\dots,n,
\end{equation}
where $A$ is a real $n\times n$ matrix and $\vect a_i=(A_{1i},\dots,A_{ni})$ is the $i$-th column of the real hierarchy matrix $A$.
Define the weighted inner product
\begin{equation}
\wex{\vect f}{\vect h}\equiv \sum_{i=1}^n w_i\,\mathbb E_{\theta_i}[f_i(x)h_i(x)].
\end{equation}
Applying the Cauchy--Schwarz inequality to $\Delta_i\hat\theta(x)\equiv \hat\theta(x)-\mathbb E_{\theta_i}[\hat\theta(x)]$ and $S_i(x)$ yields
\begin{equation}\label{eq:SM_CS_ineq}
\wex{\Delta\hat{\vect\theta}}{\vect S}^2
\le
\wex{\Delta\hat{\vect\theta}}{\Delta\hat{\vect\theta}}\,
\wex{\vect S}{\vect S}.
\end{equation}

Using the definition of the generalized global score function~\eqref{eq:SM_score_func}, the left-hand side becomes
\begin{equation}\label{eq:SM_lhs_eval}
\wex{\Delta\hat{\vect\theta}}{\vect S}
=
\sum_{i=1}^n w_i \sum_{x\in X} p(x|\theta_i)\Delta_i\hat\theta(x)S_i(x)=
\sum_{i,j=1}^n
\left[
w_i A_{ji}
\left(
\sum_{x\in X}\Delta_i\hat\theta(x)g_j(x)
\right)
\right]
=
\sum_{i=1}^n w_i \vect a_i^\top\vect b^{(i)},
\end{equation}
where $\vect b^{(i)}$ is the $i$-column of the bias matrix defined by
\begin{equation}\label{eq:SM_bias_matrix}
\vect b^{(i)} \equiv 
\sum_{x\in X} \vect g(x)\,\Delta_i\hat\theta(x)=\sum_{x\in X} \vect g(x)\left[\hat\theta(x)-\mathbb E_{\theta_i}(\hat\theta)\right]
\end{equation}
The first factor on the right-hand side of Eq.~\eqref{eq:SM_CS_ineq} evaluates to the weighted variance,
\begin{equation}\label{eq:SM_rhs_first_eval}
\wex{\Delta\hat{\vect\theta}}{\Delta\hat{\vect\theta}}
=
\sum_{i=1}^n w_i\,\mathbb E_{\theta_i}\!\left[\Delta_i\hat\theta(x)^2\right]
=
\overline{\var}(\hat\theta).
\end{equation}
The second factor becomes
\begin{align}\label{eq:SM_rhs_second_eval}
\wex{\vect S}{\vect S}
&=
\sum_{i=1}^n w_i \sum_{x\in X} p(x|\theta_i)\,[S_i(x)]^2=
\sum_{i=1}^n w_i \sum_{x\in X}
\frac{\left[\sum_{j=1}^n A_{ji}g_j(x)\right]^2}{p(x|\theta_i)}
=
\sum_{i=1}^n w_i \vect a_i^\top\mathcal C^{(i)}\vect a_i,
\end{align}
where $\mathcal C^{(i)}$ is the classical information matrix with entries
\begin{equation}\label{eq:SM_CIM}
[\mathcal C^{(i)}]_{jk}
=
\sum_{x\in X}\frac{g_j(x)g_k(x)}{p(x|\theta_i)}.
\end{equation}

Substituting Eqs.~\eqref{eq:SM_lhs_eval}, \eqref{eq:SM_rhs_first_eval}, and \eqref{eq:SM_rhs_second_eval} into the Cauchy--Schwarz inequality~\eqref{eq:SM_CS_ineq} yields a general bound for the weighted variance,
\begin{equation}\label{eq:SM_general_bound}
\overline{\var}(\hat\theta)
\ge
\frac{
\left(\sum_i w_i\,\mathbf a_i^\top\mathbf b^{(i)}\right)^2
}{
\sum_i w_i\,\mathbf a_i^\top\mathcal C^{(i)}\mathbf a_i
}
\equiv
\mathcal B_{\rm C}(A),
\end{equation}
which is the general global-score bound obtained in the main text.
This bound holds for any choice of test functions $\vect g(x)$ and any hierarchy matrix $A$ for which the denominator is nonzero.
The specific bounds presented in the main text are obtained by optimizing $\BC(A)$ under different structural constraints on $A$.

\subsection{Optimization over the Hierarchy Matrix}
\label{SM_Optimizatio_A}

The bound $\BC(A)$ in Eq.~\eqref{eq:SM_general_bound} depends on the hierarchy matrix $A$.
To obtain the strongest bound, we optimize over $A$.
Before performing this optimization, we introduce the following mathematical result.

\begin{theorem}[Rayleigh Quotient Maximization]
\label{thm:rayleigh_max}
Let $\vect\mu,\vect x\in\mathbb R^k$, and let $M$ be a $k\times k$ real symmetric positive semidefinite matrix.
Define the generalized Rayleigh quotient
\begin{equation}\label{SM:Rayleigh_quotient}
R(\vect x;\vect\mu,M)=\frac{(\vect\mu^\top\vect x)^2}{\vect x^\top M\vect x}.
\end{equation}
The maximum over $\vect x\neq 0$ is
\begin{equation}\label{SM:Rayleigh_quotient_max}
R_{\max}(\vect\mu,M)
=
\max_{\vect x\in\mathbb R^k\setminus\{0\}}R(\vect x;\vect\mu,M)
=
\begin{cases}
\vect\mu^\top M^+\vect\mu, & \text{if } \vect\mu\in\operatorname{Range}(M),\\
+\infty, & \text{if } \vect\mu\notin\operatorname{Range}(M),
\end{cases}
\end{equation}
where $M^+$ denotes the Moore--Penrose pseudoinverse and
$\operatorname{Range}(M)=\{M\vect y\mid \vect y\in\mathbb R^k\}$ is the range of $M$.
Moreover, the maximum is achieved at $\vect x=\lambda M^+\vect\mu+\vect z$ with $\lambda\in\mathbb R$ and $\vect z\in\ker(M)$.
\end{theorem}

\begin{proof}
We distinguish two cases.

\emph{Case 1:} $\vect\mu\notin\operatorname{Range}(M)$.
Since $M$ is symmetric, $\operatorname{Range}(M)=(\ker M)^\perp$.
Thus $\vect\mu\notin\operatorname{Range}(M)$ implies that $\vect\mu$ has a nonzero component in $\ker(M)$.
Choose $\vect x\in\ker(M)$ such that $\vect\mu^\top\vect x\neq 0$.
Then $\vect x^\top M\vect x=0$ while $(\vect\mu^\top\vect x)^2>0$, so the denominator of Eq.~\eqref{SM:Rayleigh_quotient} vanishes.
By rescaling $\vect x$, the quotient can be made arbitrarily large, and hence the supremum is $+\infty$.

\emph{Case 2:} $\vect\mu\in\operatorname{Range}(M)$.
In this case the Rayleigh quotient is finite.
Since the quotient is homogeneous in $\vect x$, we may impose the normalization constraint
\begin{equation}
\vect x^\top M\vect x=1,
\end{equation}
and maximize $(\vect\mu^\top\vect x)^2$ under this constraint.

Consider the Lagrangian
\begin{equation}
\mathcal L(\vect x,\lambda)=\vect\mu^\top\vect x-\lambda(\vect x^\top M\vect x-1),
\end{equation}
where maximizing $\vect\mu^\top\vect x$ is sufficient because the square in Eq.~\eqref{SM:Rayleigh_quotient} removes the sign.
The stationary condition $\nabla_{\vect x}\mathcal L=0$ yields
\begin{equation}
\vect\mu=2\lambda M\vect x.
\end{equation}
Applying the pseudoinverse gives
\begin{equation}
\vect x=\frac{1}{2\lambda}M^+\vect\mu+\vect v,\qquad \vect v\in\ker(M).
\end{equation}
Because $\vect\mu\in\operatorname{Range}(M)=(\ker M)^\perp$, we have $\vect\mu^\top\vect v=0$.
Furthermore, $M\vect v=0$, so $\vect v$ does not affect either the numerator or the constraint.
Hence we can set $\vect v=0$ without loss of generality, giving
\begin{equation}\label{SM:THM1_1}
\vect x=\frac{1}{2\lambda}M^+\vect\mu.
\end{equation}
Substituting into the constraint $\vect x^\top M\vect x=1$ gives
\begin{equation}
\frac{1}{(2\lambda)^2}\vect\mu^\top M^+MM^+\vect\mu
=
\frac{1}{(2\lambda)^2}\vect\mu^\top M^+\vect\mu
=1,
\end{equation}
where we used $M^+MM^+=M^+$.
Therefore,
\begin{equation}\label{SM:THM1_2}
2\lambda=\pm\sqrt{\vect\mu^\top M^+\vect\mu}.
\end{equation}
Substituting Eqs.~\eqref{SM:THM1_1} and \eqref{SM:THM1_2} into Eq.~\eqref{SM:Rayleigh_quotient} yields
\begin{equation}
R_{\max}(\vect\mu,M)=\vect\mu^\top M^+\vect\mu.
\end{equation}
From Eq.~\eqref{SM:THM1_1}, the maximum is achieved for $\vect x=\lambda M^+\vect\mu+\vect z$ with $\lambda\in\mathbb R$ and $\vect z\in\ker(M)$.
\end{proof}

\textbf{Diagonal hierarchy matrix}.--
For a diagonal hierarchy matrix
\begin{equation}
A_{\rm diag}=\operatorname{diag}(a_1,\dots,a_n),
\end{equation}
Eq.~\eqref{eq:SM_general_bound} simplifies to
\begin{equation}
\BC(A_{\rm diag})
=
\frac{\left(\sum_{i=1}^n w_i a_i  b_{i}^{(i)}\right)^2}
{\sum_{i=1}^n w_i a_i^2[\mathcal C^{(i)}]_{ii}}.
\end{equation}
This corresponds to the Rayleigh quotient~\eqref{SM:Rayleigh_quotient} with
\begin{equation}
\vect x=(a_1,\dots,a_n)^\top,\qquad
\vect\mu=(w_1 b_{1}^{(1)},\dots,w_n  b_{n}^{(n)})^\top,
\end{equation}
and
\begin{equation}
M=\operatorname{diag}\!\bigl(w_1[\mathcal C^{(1)}]_{11},\dots,w_n[\mathcal C^{(n)}]_{nn}\bigr).
\end{equation}
Since $\vect\mu\in\operatorname{Range}(M)$ holds trivially for this diagonal structure, Theorem~\ref{thm:rayleigh_max} gives
\begin{equation}\label{SM_classical_diagonal_bound}
\BCd
=
\max_{A_{\rm diag}}\BC(A_{\rm diag})
=
\vect\mu^\top M^+\vect\mu
=
\sum_{i=1}^n w_i\frac{( b_{i}^{(i)})^2}{[\mathcal C^{(i)}]_{ii}}.
\end{equation}
This is the classical bound corresponding to the global Cram\'er--Rao level in the main text.

\textbf{Column-constant hierarchy matrix}.--
For a column-constant hierarchy matrix
\begin{equation}
A_{\rm col}=(\vect a,\dots,\vect a),
\end{equation}
i.e. $\vect a_{i}=\vect a$ for all $i$, Eq.~\eqref{eq:SM_general_bound} simplifies to
\begin{equation}
\BC(A_{\rm col})
=
\frac{(\vect b_w^\top\vect a)^2}{\vect a^\top\mathcal C_w\vect a},
\end{equation}
where
\begin{equation}
\vect b_w=\sum_{i=1}^n w_i\vect b^{(i)},\qquad
\mathcal C_w=\sum_{i=1}^n w_i\mathcal C^{(i)}.
\end{equation}
This is again a Rayleigh quotient with
\begin{equation}
\vect x=\vect a,\qquad \vect\mu=\vect b_w,\qquad M=\mathcal C_w.
\end{equation}
The condition $\vect\mu\in\operatorname{Range}(M)$ required for a finite maximum is ensured if $\vect b^{(i)}\in\operatorname{Range}(\mathcal C^{(i)})$ for all $i$, which will be proven in the next subsection.
Assuming this condition holds, Theorem~\ref{thm:rayleigh_max} gives
\begin{equation}
\BCc
=
\max_{A_{\rm col}}\BC(A_{\rm col})
=
\vect b_w^\top \mathcal C_w^+ \vect b_w.
\end{equation}
This is the classical bound corresponding to the global Barankin level in the main text.

\textbf{General hierarchy matrix}.--
For the fully unrestricted case, we impose no constraint on the hierarchy matrix $A$.
To express Eq.~\eqref{eq:SM_general_bound} as a Rayleigh quotient, we vectorize all columns of $A$.
Let
\begin{equation}
\vect x=(\vect a_{1}^{\top},\vect a_2^{\top},\dots,\vect a_n^{\top})^\top\in\mathbb R^{n^2}
\end{equation}
be the concatenation of all column vectors.
Similarly, define
\begin{equation}
\vect\mu=(w_1\vect b^{(1)\top},w_2\vect b^{(2)\top},\dots,w_n\vect b^{(n)\top})^\top\in\mathbb R^{n^2}.
\end{equation}
The denominator in Eq.~\eqref{eq:SM_general_bound} becomes a quadratic form $\vect x^\top M\vect x$, where $M$ is the block-diagonal matrix
\begin{equation}
M=\bigoplus_{i=1}^n\bigl(w_i\mathcal C^{(i)}\bigr)\in\mathbb R^{n^2\times n^2}.
\end{equation}
With these definitions,
\begin{equation}
\BC(A)=\frac{(\vect\mu^\top\vect x)^2}{\vect x^\top M\vect x},
\end{equation}
which is precisely the Rayleigh quotient in Theorem~\ref{thm:rayleigh_max}.
The condition for a finite maximum, $\vect\mu\in\operatorname{Range}(M)$, is equivalent to requiring $\vect b^{(i)}\in\operatorname{Range}(\mathcal C^{(i)})$ for each $i$.
Assuming this holds, Theorem~\ref{thm:rayleigh_max} yields
\begin{equation}\label{SM_classical_gen_bound}
\BCg
=
\max_A \BC(A)
=
\vect\mu^\top M^+\vect\mu
=
\sum_{i=1}^n w_i\,\vect b^{(i)\top}(\mathcal C^{(i)})^+\vect b^{(i)}.
\end{equation}
This is the classical bound corresponding to the fully global level in the main text.

\subsection{Automatic Finite Condition}
\label{SM_finite_condition}

We first prove the finite condition
\begin{equation}\label{eq:SM_finite_condition}
\vect b^{(i)}\in \operatorname{Range}(\mathcal C^{(i)}),\qquad \forall i=1,\dots,n,
\end{equation}
which ensures that the classical bounds remain finite, i.e., that the corresponding Rayleigh quotient takes a finite value.
We emphasize that this condition is always satisfied.

For each $i$, define the subspace
\begin{equation}
\mathcal G\equiv \operatorname{Span}\{\vect g(x)\mid x\in X\}\subseteq \mathbb R^n.
\end{equation}
According to the definition of the bias vector in Eq.~\eqref{eq:SM_bias_matrix}, we have
\begin{equation}
\vect b^{(i)}=\sum_{x\in X}\vect g(x)\,\Delta\hat\theta_i(x),
\end{equation}
which is a linear combination of the vectors $\vect g(x)$.
Hence $\vect b^{(i)}\in\mathcal G$ for any estimator.
Thus, to establish the finite condition, it suffices to show that
\begin{equation}\label{eq:SM_im_C}
\operatorname{Range}(\mathcal C^{(i)})=\mathcal G.
\end{equation}

We prove the two inclusions separately.
First, let $\vect u\in \operatorname{Range}(\mathcal C^{(i)})$.
Then there exists $\vect v\in\mathbb R^n$ such that $\vect u=\mathcal C^{(i)}\vect v$.
Substituting the definition of the classical information matrix from Eq.~\eqref{eq:SM_CIM}, we obtain
\begin{equation}
\vect u
=
\mathcal C^{(i)}\vect v
=
\sum_{x\in X}\frac{\vect g(x)\,[\vect g(x)^\top\vect v]}{p(x|\theta_i)},
\end{equation}
which is a linear combination of the vectors $\vect g(x)$.
Therefore $\vect u\in\mathcal G$, establishing $\operatorname{Range}(\mathcal C^{(i)})\subseteq \mathcal G$.

For the reverse inclusion, we use orthogonal complements.
Let $\vect\mu\in (\operatorname{Range}\mathcal C^{(i)})^\perp$, meaning that
\begin{equation}
\vect\mu^\top \mathcal C^{(i)} \vect\nu =0,\qquad \forall \vect\nu\in\mathbb R^n.
\end{equation}
In particular, taking $\vect\nu=\vect\mu$ gives
$\vect\mu^\top \mathcal C^{(i)} \vect\mu =0$.
Using Eq.~\eqref{eq:SM_CIM}, this becomes
\begin{equation}
\vect\mu^\top \mathcal C^{(i)} \vect\mu
=
\sum_{x\in X}\frac{[\vect\mu^\top\vect g(x)]^2}{p(x|\theta_i)}
=0.
\end{equation}
Since $p(x|\theta_i)>0$ for all $x\in X$, each term in the sum is nonnegative.
Hence every term must vanish, which implies
\begin{equation}
\vect\mu^\top \vect g(x)=0,\qquad \forall x\in X.
\end{equation}
Thus $\vect\mu$ is orthogonal to every $\vect g(x)$, and therefore $\vect\mu\in\mathcal G^\perp$.
We have shown that $(\operatorname{Range}\mathcal C^{(i)})^\perp \subseteq \mathcal G^\perp$.
Taking orthogonal complements on both sides yields $\operatorname{Range}(\mathcal C^{(i)})\supseteq \mathcal G$.

Combining the two inclusions, we conclude that $\operatorname{Range}(\mathcal C^{(i)})=\mathcal G$,
as claimed in Eq.~\eqref{eq:SM_im_C}.
Since $\vect b^{(i)}\in\mathcal G$ for any estimator, it follows that $\vect b^{(i)}\in\operatorname{Range}(\mathcal C^{(i)})$.
Thus, the finite condition~\eqref{eq:SM_finite_condition} holds automatically for all $i$, without any additional assumptions.

\section{Attainability of the Fully Global Classical Bound}

\subsection{Saturation Condition}

We now derive the saturation condition for the classical fully global bound, which is tighter than the diagonal and column-constant bounds.

In deriving the fully global bound, we first apply the Cauchy--Schwarz inequality~\eqref{eq:SM_CS_ineq}, which yields
\begin{equation}\label{eq:SM_saturation_cond_1}
\Delta_i\hat\theta(x)
=
\lambda_i S_i(x)
=
\lambda_i\,\frac{\vect g(x)^\top \vect a_i}{p(x|\theta_i)},
\qquad
\forall i:\; w_i\neq 0,
\end{equation}
for some $\lambda_i\in\mathbb R$.
Optimizing over the hierarchy matrix $A$ with Theorem~\ref{thm:rayleigh_max} then gives
\begin{equation}\label{eq:SM_saturation_cond_2}
\vect a_i
=
\beta_i(\mathcal C^{(i)})^+\vect b^{(i)}+\vect z^{(i)},
\qquad
\forall i:\; w_i\neq 0,
\end{equation}
for some $\beta_i\in\mathbb R$ and $\vect z^{(i)}\in\ker(\mathcal C^{(i)})$.
Combining Eqs.~\eqref{eq:SM_saturation_cond_1} and \eqref{eq:SM_saturation_cond_2}, we obtain
\begin{equation}\label{eq:SM_saturation_cond_comb}
\vect g(x)^\top(\mathcal C^{(i)})^+\vect b^{(i)}
+
\vect g(x)^\top\vect z^{(i)}
=
\alpha_i\,p(x|\theta_i)\,\Delta_i\hat\theta(x),
\qquad
\forall i:\; w_i\neq 0,
\end{equation}
for some $\alpha_i\in\mathbb R$ and $\vect z^{(i)}\in\ker(\mathcal C^{(i)})$.

The constants $\alpha_i$ are fixed by substituting Eq.~\eqref{eq:SM_saturation_cond_comb} into the bias definition~\eqref{eq:SM_bias_matrix}. For each $i$,
\begin{align}
\vect b^{(i)}
&=
\sum_{x\in X}\vect g(x)\,\Delta_i\hat\theta(x)
=
\sum_{x\in X}
\vect g(x)\,
\frac{\vect g(x)^\top\!\left[(\mathcal C^{(i)})^+\vect b^{(i)}+\vect z^{(i)}\right]}
{\alpha_i\,p(x|\theta_i)}
=
\frac{\mathcal C^{(i)}\left[(\mathcal C^{(i)})^+\vect b^{(i)}+\vect z^{(i)}\right]}{\alpha_i}.
\end{align}
Since $\vect b^{(i)}\in\operatorname{Range}(\mathcal C^{(i)})$ by Sec.~\ref{SM_finite_condition} and $\vect z^{(i)}\in\ker(\mathcal C^{(i)})$, we have $\mathcal C^{(i)}(\mathcal C^{(i)})^+\vect b^{(i)}=\vect b^{(i)}$, hence $\alpha_i=1$.
We thus obtain the saturation condition
\begin{equation}\label{eq:SM_saturation_cond_with_kernel}
\vect g(x)^\top\!\left[(\mathcal C^{(i)})^+\vect b^{(i)}+\vect z^{(i)}\right]
=
p(x|\theta_i)\,\Delta_i\hat\theta(x),
\qquad
\forall i:\; w_i\neq 0,
\end{equation}
for some $\vect z^{(i)}\in\ker(\mathcal C^{(i)})$, where $\Delta_i\hat\theta(x)=\hat\theta(x)-\mathbb E_{\theta_i}(\hat\theta)$.

The kernel term does not affect the outcome-space equation. Indeed, if $\vect z^{(i)}\in\ker(\mathcal C^{(i)})$, then
\begin{equation}
0
=
\vect z^{(i)\top}\mathcal C^{(i)}\vect z^{(i)}
=
\sum_{x\in X}
\frac{\left[\vect g(x)^\top\vect z^{(i)}\right]^2}{p(x|\theta_i)}.
\end{equation}
Under the regularity assumption $p(x|\theta_i)>0$ for all $x\in X$, this implies
\begin{equation}
\vect g(x)^\top\vect z^{(i)}=0,
\qquad
\forall x\in X.
\end{equation}
Hence the saturation condition is equivalently
\begin{equation}\label{eq:SM_saturation_cond_final}
\vect g(x)^\top(\mathcal C^{(i)})^+\vect b^{(i)}
=
p(x|\theta_i)\,\Delta_i\hat\theta(x),
\qquad
\forall i:\; w_i\neq0.
\end{equation}

To verify sufficiency, we substitute Eq.~\eqref{eq:SM_saturation_cond_final} into Eq.~\eqref{eq:weighted_var_SM}:
\begin{align}
\overline{\var}(\hat\theta)
&=
\sum_{i=1}^n w_i\sum_{x\in X}p(x|\theta_i)[\Delta_i\hat\theta(x)]^2
=
\sum_{i=1}^n w_i\sum_{x\in X}
\frac{1}{p(x|\theta_i)}
\left[\vect b^{(i)\top}(\mathcal C^{(i)})^+\vect g(x)\right]^2=
\sum_{i=1}^n w_i\,\vect b^{(i)\top}(\mathcal C^{(i)})^+\vect b^{(i)},
\end{align}
which is precisely the classical fully global bound $\BCg$.

\subsection{Existence of Saturating Estimators}

The saturation condition~\eqref{eq:SM_saturation_cond_final} should be understood as a set of compatibility equations for a single estimator $\hat\theta:X\to\mathbb R$. Although the finite condition $\vect b^{(i)}\in\operatorname{Range}(\mathcal C^{(i)})$ holds automatically, this does not by itself guarantee the existence of a common estimator $\hat\theta(x)$ satisfying Eq.~\eqref{eq:SM_saturation_cond_final} for all parameter points with nonzero weights.

It is useful to rewrite this condition in a purely linear-algebraic form. Let $|X|$ denote the number of outcomes and introduce the outcome vector
\begin{equation}
\vect h=(\hat\theta(x))_{x\in X}\in\mathbb R^{|X|},
\end{equation}
the probability vector
\begin{equation}
\vect p_i=(p(x|\theta_i))_{x\in X}\in\mathbb R^{|X|},
\end{equation}
and the diagonal matrix
\begin{equation}
P_i=\operatorname{diag}(p(x|\theta_i))_{x\in X}.
\end{equation}
Furthermore, define the matrix
\begin{equation}
G=\big(\vect g(x)\big)_{x\in X}\in\mathbb R^{n\times |X|},
\end{equation}
whose columns are the vectors $\vect g(x)$. Then
\begin{equation}
\mathbb E_{\theta_i}(\hat\theta)=\vect p_i^\top\vect h,
\end{equation}
and Eq.~\eqref{eq:SM_saturation_cond_final} implies
\begin{equation}\label{eq:SM_estimator_existence_condition}
P_i\left(\vect h-\vect 1\,\vect p_i^\top\vect h\right)
\in
\operatorname{Range}(G^\top),
\qquad
\forall i:\; w_i\neq0.
\end{equation}

Conversely, if Eq.~\eqref{eq:SM_estimator_existence_condition} holds, then for each $i$ there exists a vector $\vect u^{(i)}\in\mathbb R^n$ such that
\begin{equation}\label{sm1}
P_i\left(\vect h-\vect 1\,\vect p_i^\top\vect h\right)=G^\top \vect u^{(i)} .
\end{equation}
Equivalently,
\begin{equation}
\Delta_i\hat\theta(x)=\frac{\vect g(x)^\top \vect u^{(i)}}{p(x|\theta_i)},
\qquad \forall x\in X.
\end{equation}
Substituting this relation into the definition of the bias vector gives
\begin{equation}
\vect b^{(i)}=\sum_x \vect g(x)\Delta_i\hat\theta(x)=\mathcal C^{(i)}\vect u^{(i)}.
\end{equation}
Hence $\vect u^{(i)}$ is a solution of $\mathcal C^{(i)}\vect u^{(i)}=\vect b^{(i)}$, and therefore can be written as
\begin{equation}\label{u1}
\vect u^{(i)}=(\mathcal C^{(i)})^+\vect b^{(i)}+\vect z^{(i)},
\qquad
\vect z^{(i)}\in\ker(\mathcal C^{(i)}).
\end{equation}
Substituting Eq.~\eqref{u1} into Eq.~\eqref{sm1}, we recover the saturation equation~\eqref{eq:SM_saturation_cond_final}.

Therefore, a necessary and sufficient condition for the existence of an estimator saturating the fully global classical bound is
\begin{equation}\label{eq:SM_general_estimator_solution_condition}
\exists\,\vect h\in\mathbb R^{|X|}
\quad\text{such that}\quad
P_i\left(\vect h-\vect 1\,\vect p_i^\top\vect h\right)
\in
\operatorname{Range}(G^\top),
\quad
\forall i:\; w_i\neq0.
\end{equation}
Equivalently, the centered fluctuation of the estimator at each parameter point, weighted by the corresponding probability distribution, must lie in the range of $G^\top$.

A useful sufficient condition follows immediately from the rank of $G$. If the vectors $\{\vect g(x)\mid x\in X\}$ are linearly independent, then the matrix $G$
has full column rank, i.e.,
\begin{equation}
\operatorname{rank}(G)=|X|.
\end{equation}
Since $\operatorname{rank}(G^\top)=\operatorname{rank}(G)$, this implies
\begin{equation}\label{eq:SM_full_column_rank_condition}
\operatorname{Range}(G^\top)=\mathbb R^{|X|}.
\end{equation}
Under this condition, Eq.~\eqref{eq:SM_general_estimator_solution_condition} is automatically satisfied for every estimator $\hat\theta(x)$. Hence the fully global classical bound is tight for arbitrary estimators.

If one additionally imposes global unbiasedness,
\begin{equation}
\mathbb E_{\theta_i}(\hat\theta)=\theta_i,
\qquad
i=1,\ldots,n,
\end{equation}
then Eq.~\eqref{eq:SM_general_estimator_solution_condition} becomes
\begin{equation}\label{eq:SM_unbiased_estimator_solution_condition}
\exists\,\vect h\in\mathbb R^{|X|}
\quad\text{such that}\quad
\vect p_i^\top\vect h=\theta_i,
\qquad
P_i(\vect h-\theta_i\vect 1)
\in
\operatorname{Range}(G^\top),
\quad
\forall i:\; w_i\neq0.
\end{equation}
This condition ensures the existence of an unbiased estimator saturating the classical fully global bound.

\subsection{Anchor-based Construction of a Globally Unbiased Estimator}

We now show how the general existence condition simplifies when one restricts to globally unbiased estimators. Assume that there exists an estimator $\hat\theta:X\to\mathbb R$ satisfying
\begin{equation}
\mathbb E_{\theta_k}[\hat\theta]=\theta_k,
\qquad
k=1,\ldots,n.
\end{equation}
Then the general compatibility condition~\eqref{eq:SM_unbiased_estimator_solution_condition} implies that, for each fixed reference point $\theta_i$,
\begin{equation}\label{eq:SM_anchor_start}
P_i(\vect h-\theta_i\vect 1)\in \operatorname{Range}(G^\top),
\end{equation}
where $\vect h=(\hat\theta(x))_{x\in X}$.

By the definition of $\operatorname{Range}(G^\top)$, Eq.~\eqref{eq:SM_anchor_start} is equivalent to the existence of a vector $\vect c^{(i)}\in\mathbb R^n$ such that
\begin{equation}
P_i(\vect h-\theta_i\vect 1)=G^\top \vect c^{(i)}.
\end{equation}
Writing this relation componentwise gives
\begin{equation}
p(x|\theta_i)\bigl(\hat\theta(x)-\theta_i\bigr)=\vect g(x)^\top \vect c^{(i)},
\qquad
\forall x\in X,
\end{equation}
or equivalently,
\begin{equation}\label{eq:SM_anchor_ansatz}
\hat\theta(x)=\theta_i+\frac{\vect g(x)^\top \vect c^{(i)}}{p(x|\theta_i)}.
\end{equation}
Thus, once global unbiasedness is imposed, any estimator satisfying the general compatibility condition can be represented in this anchor-based form.

We now determine the corresponding condition on $\vect c^{(i)}$. Substituting Eq.~\eqref{eq:SM_anchor_ansatz} into the global unbiasedness condition
\begin{equation}
\mathbb E_{\theta_k}[\hat\theta]=\theta_k,
\qquad
k=1,\ldots,n,
\end{equation}
yields
\begin{equation}
\theta_i+\sum_x p(x|\theta_k)\frac{\vect g(x)^\top \vect c^{(i)}}{p(x|\theta_i)}=\theta_k,
\qquad
k=1,\ldots,n.
\end{equation}
To express this compactly, define the parameter-displacement vector
\begin{equation}
\Delta\vect\theta^{(i)}
=
\begin{pmatrix}
\theta_1-\theta_i\\
\vdots\\
\theta_n-\theta_i
\end{pmatrix},
\end{equation}
and the matrix
\begin{equation}
[M^{(i)}]_{jk}
=
\sum_{x\in X}\frac{p(x|\theta_j)}{p(x|\theta_i)}\,g_k(x).
\end{equation}
Then the global unbiasedness equations reduce to
\begin{equation}\label{eq:SM_anchor_linear_system}
\Delta\vect\theta^{(i)}=M^{(i)}\vect c^{(i)}.
\end{equation}
Therefore, within the anchor-based representation, the existence of a globally unbiased estimator is equivalent to
\begin{equation}\label{eq:SM_anchor_image_condition}
\Delta\vect\theta^{(i)}\in\operatorname{Range}(M^{(i)}).
\end{equation}
This is precisely the general existence condition rewritten in anchor-based form. It provides a more explicit criterion for globally unbiased estimators by replacing the abstract compatibility equation with an range condition for the matrix $M^{(i)}$.

Finally, to connect this construction with saturation of the fully global classical bound, note that Eq.~\eqref{eq:SM_anchor_ansatz} implies
\begin{equation}
\Delta\hat\theta_i(x)=\frac{\vect g(x)^\top\vect c^{(i)}}{p(x|\theta_i)},
\end{equation}
which has the same structure as the saturation equation~\eqref{eq:SM_saturation_cond_final}. Substituting this into the definition of the bias vector gives
\begin{equation}
\vect b^{(i)}=\sum_x \vect g(x)\Delta\hat\theta_i(x)=\mathcal C^{(i)}\vect c^{(i)}.
\end{equation}
Hence the anchor-based representation saturates the fully global classical bound whenever
\begin{equation}\label{eq:SM_anchor_saturation_condition}
\vect b^{(i)}=\mathcal C^{(i)}\vect c^{(i)}.
\end{equation}
Therefore, for globally unbiased estimators, the anchor-based construction provides a constructive reformulation of the general compatibility condition, and Eq.~\eqref{eq:SM_anchor_saturation_condition} gives the additional requirement for saturation of the fully global classical bound.

\section{Observation~1: Asymptotic Reduction to the Global Cram\'er--Rao Bound}

We begin with two preliminary lemmas and then turn to the proof of Observation~1.

\begin{lemma}\label{SM_obs1_lem1}
Let $p(x|\theta_1)$ and $p(x|\theta_2)$ be two probability distributions such that $p(x|\theta_1)>0$ whenever $p(x|\theta_2)>0$. Then
\begin{equation}
\sum_{x\in X} \frac{p(x|\theta_2)^2}{p(x|\theta_1)} \ge 1,
\end{equation}
with equality if and only if $p(x|\theta_2)=p(x|\theta_1)$ for all $x$.
\end{lemma}

\begin{proof}
Define $Y(x)=p(x|\theta_2)/p(x|\theta_1)$, so that
\begin{equation}
\sum_{x\in X} \frac{p(x|\theta_2)^2}{p(x|\theta_1)}
=
\sum_{x\in X} p(x|\theta_1)Y(x)^2
=
\mathbb E_{\theta_1}[Y^2].
\end{equation}
Since
\begin{equation}
\mathbb E_{\theta_1}[Y]
=
\sum_{x\in X} p(x|\theta_1)Y(x)
=
\sum_x p(x|\theta_2)
=
1,
\end{equation}
the non-negativity of the variance gives
\begin{equation}
\mathbb E_{\theta_1}[Y^2]\ge \bigl(\mathbb E_{\theta_1}[Y]\bigr)^2=1.
\end{equation}
Equality holds if and only if $Y$ is constant, which implies $Y(x)=1$ for all $x$, i.e., $p(x|\theta_2)=p(x|\theta_1)$.
\end{proof}

\begin{lemma}\label{SM_obs1_lem2}
Consider the block matrix
\begin{equation}\label{SM_obs1_Lemm2}
M=
\begin{pmatrix}
D & \vect v \\
\vect v^\top & c
\end{pmatrix},
\end{equation}
where $D\in\mathbb R^{(n-1)\times(n-1)}$ is symmetric, $\vect v\in\mathbb R^{n-1}$, and $c\in\mathbb R$. Assume that $\vect v\in\operatorname{Range}(D)$, $\|D\|\to\infty$, and $\|\vect v\|=O(1)$, while $c=O(1)$. Then the Moore--Penrose pseudoinverse satisfies
\begin{equation}
\lim_{\|D\|\to\infty} M^+
=
\begin{pmatrix}
0_{(n-1)\times (n-1)} & \vect 0 \\
\vect 0^\top & 1/c
\end{pmatrix}.
\end{equation}
\end{lemma}

\begin{proof}
Let $s\equiv c-\vect v^\top D^+\vect v\neq 0$ denote the generalized Schur complement. According to Ref.~\cite{burns1974generalized}, the pseudoinverse of $M$ is
\begin{equation}\label{SM_obs1_M+_0}
M^+
=
\begin{pmatrix}
D^+ + D^+\vect v\vect v^\top D^+/s & -D^+\vect v/s \\
-\vect v^\top D^+/s & 1/s
\end{pmatrix}.
\end{equation}
Under the asymptotic scaling $\|D\|\to\infty$ with $\|\vect v\|=O(1)$ and $c=O(1)$, we have $D^+\simeq O(\|D\|^{-1})$, $D^+\vect v\simeq O(\|D\|^{-1})$, and $s\simeq c+O(\|D\|^{-1})$. Hence
\begin{equation}
M^+
\simeq
\begin{pmatrix}
O(\|D\|^{-1}) & O(\|D\|^{-1}) \\
O(\|D\|^{-1}) & 1/c + O(\|D\|^{-1})
\end{pmatrix},
\end{equation}
so that the only nonvanishing entry in the limit is the $(n,n)$ component,
\begin{equation}
[M^+]_{nn}\to \frac{1}{c}\qquad \text{as}\qquad \|D\|\to\infty.
\end{equation}
\end{proof}

\textbf{Observation~1.}
For finite parameter dimension $n$, the fully global bound reduces to the global Cram\'er--Rao bound in the limit of many independent repetitions $m\gg1$,
\begin{equation}\label{eq:observation_1}
\overline{\var}(\hat\theta)
\ge
\mathcal B_{\rm C}^{\rm FG}
\overset{m\gg1}{\simeq}
\mathcal B_{\rm C}^{\rm GCR},
\end{equation}
when the test functions are chosen as $g_j(x)=\partial_\theta p(x|\theta_j)$ for all $j$ and $\{p(x|\theta_j)\}$ is not identical.
Moreover, this asymptotic bound is saturated by the maximum-likelihood estimator (MLE).

\begin{proof}
We first prove that $\BCg\simeq \BCd$ as $m\to\infty$, and then show that the maximum-likelihood estimator saturates both bounds in this limit.

Consider the classical information matrix for $m$ independent measurement repetitions at the parameter value $\theta_n$. It is given by
\begin{equation}
[\mathcal C^{(n)}(m)]_{jk}
=
\sum_{\vec x}
\frac{\partial_\theta p(\vec x|\theta_j)\,\partial_\theta p(\vec x|\theta_k)}
     {p(\vec x|\theta_n)},
\end{equation}
where $\vec x=(x_1,\dots,x_m)$ denotes the outcomes of the $m$ repetitions, so that $p(\vec x|\theta)=\prod_{r=1}^m p(x_r|\theta)$.
Define $s_j(x)=\partial_\theta\log p(x|\theta_j)$. Then
\begin{equation}
\partial_\theta p(\vec x|\theta_j)=p(\vec x|\theta_j)\sum_{r=1}^m s_j(x_r),
\end{equation}
which gives
\begin{align}\label{SM_Obs_1_C1}
[\mathcal C^{(n)}(m)]_{jk}
&=
\sum_{\vec x}
\frac{p(\vec x|\theta_j)p(\vec x|\theta_k)}{p(\vec x|\theta_n)}
\left(\sum_{r=1}^m s_j(x_r)\right)
\left(\sum_{t=1}^m s_k(x_t)\right)\nonumber\\
&=
\sum_{r,t=1}^m
\left(
\sum_{\vec x}
p(\vec x|\theta_n)
\frac{p(\vec x|\theta_j)p(\vec x|\theta_k)}{p(\vec x|\theta_n)^2}
s_j(x_r)s_k(x_t)
\right).
\end{align}

Let $a_j(x)=p(x|\theta_j)/p(x|\theta_n)$. Then
\begin{equation}
\frac{p(\vec x|\theta_j)p(\vec x|\theta_k)}{p(\vec x|\theta_n)^2}
=
\prod_{u=1}^m a_j(x_u)a_k(x_u).
\end{equation}
Substituting this into Eq.~\eqref{SM_Obs_1_C1}, we obtain
\begin{equation}\label{SM_obs1_expect}
[\mathcal C^{(n)}(m)]_{jk}
=
\sum_{r,t=1}^m
\mathbb E_{\theta_n}\!\left[
s_j(x_r)s_k(x_t)\prod_{u=1}^m a_j(x_u)a_k(x_u)
\right].
\end{equation}

Introduce the single-shot quantities
\begin{align}
\Lambda_{jk}
&\equiv
\sum_{x} p(x|\theta_n)a_j(x)a_k(x)
=
\sum_x \frac{p(x|\theta_j)p(x|\theta_k)}{p(x|\theta_n)},
\label{SM_obs1_imp_1}\\
R_{jk}
&\equiv
\sum_{x} p(x|\theta_n)s_j(x)a_j(x)a_k(x)
=
\sum_x \frac{p(x|\theta_k)\partial_\theta p(x|\theta_j)}{p(x|\theta_n)}.
\label{SM_obs1_imp_2}
\end{align}
For $r\neq t$, the expectation in Eq.~\eqref{SM_obs1_expect} factorizes, giving $\Lambda_{jk}^{m-2}R_{jk}R_{kj}$.
For $r=t$, it gives $\Lambda_{jk}^{m-1}[\mathcal C^{(n)}(1)]_{jk}$, where
\begin{equation}
[\mathcal C^{(n)}(1)]_{jk}
=
\sum_{x} p(x|\theta_n)s_j(x)s_k(x)a_j(x)a_k(x)
=
\sum_{x}\frac{\partial_\theta p(x|\theta_j)\,\partial_\theta p(x|\theta_k)}{p(x|\theta_n)}.
\end{equation}
Collecting the $m(m-1)$ off-diagonal terms and the $m$ diagonal terms yields
\begin{equation}\label{SM_obs1_Cm}
[\mathcal C^{(n)}(m)]_{jk}
=
m(m-1)\Lambda_{jk}^{m-2}R_{jk}R_{kj}
+
m\Lambda_{jk}^{m-1}[\mathcal C^{(n)}(1)]_{jk}.
\end{equation}

To analyze the block structure, set $j,k=1,\dots,n-1$ and define
\begin{equation}
D_{jk}=\frac{[\mathcal C^{(n)}(m)]_{jk}}{m},\qquad
v_j=\frac{[\mathcal C^{(n)}(m)]_{jn}}{m},\qquad
c=\frac{[\mathcal C^{(n)}(m)]_{nn}}{m}.
\end{equation}
Equation~\eqref{SM_obs1_Cm} then takes the form
\begin{equation}
\frac{\mathcal C^{(n)}(m)}{m}
=
\begin{pmatrix}
D & \vect v\\
\vect v^\top & c
\end{pmatrix}.
\end{equation}

Using $R_{jn}=\sum_{x}\partial_\theta p(x|\theta_j)=0$ and $\Lambda_{jn}=\Lambda_{nj}=\sum_{x}p(x|\theta_j)=1$, we find from Eq.~\eqref{SM_obs1_Cm} that $c=[\mathcal C^{(n)}(1)]_{nn}=O(1)$ and $v_j=[\mathcal C^{(n)}(1)]_{jn}=O(1)$.
Since the distributions $p(x|\theta_j)$ are not all identical, there exists some index $i$ such that $p(x|\theta_i)\neq p(x|\theta_n)$. Lemma~\ref{SM_obs1_lem1} then implies $\Lambda_{ii}>1$, and Eq.~\eqref{SM_obs1_Cm} gives $[\mathcal C^{(n)}(m)]_{ii}/m\to\infty$ as $m\to\infty$. Consequently, $\|D\|\to\infty$.

The condition $\vect v\in\operatorname{Range}(D)$ required in Lemma~\ref{SM_obs1_lem2} is satisfied because $\operatorname{Range}(D)=\operatorname{Span}\{\vect p(\vec x)\}$ with $[\vect p(\vec x)]_j=p(\vec x|\theta_j)$ for $j=1,\dots,n-1$, as follows from Eq.~\eqref{eq:SM_im_C}. Applying Lemma~\ref{SM_obs1_lem2}, we obtain
\begin{equation}
\lim_{m\to\infty}\left(\frac{\mathcal C^{(n)}(m)}{m}\right)^+
=
\begin{pmatrix}
0_{(n-1)\times(n-1)} & \vect 0\\
\vect 0^\top & 1/[\mathcal C^{(n)}(1)]_{nn}
\end{pmatrix}.
\end{equation}
Since $(A/m)^+=mA^+$ for any scalar $m$, it follows that
\begin{equation}
[\mathcal C^{(n)}(m)]^+
\simeq
\frac{1}{m}
\begin{pmatrix}
0_{(n-1)\times(n-1)} & \vect 0\\
\vect 0^\top & 1/[\mathcal C^{(n)}(1)]_{nn}
\end{pmatrix}.
\end{equation}
Therefore,
\begin{equation}
\vect b^\top[\mathcal C^{(n)}(m)]^+\vect b
\simeq
\frac{b_n^2}{m[\mathcal C^{(n)}(1)]_{nn}}.
\end{equation}
Applying the same reasoning to each $[\mathcal C^{(i)}(m)]^+$ and substituting into the fully global bound~\eqref{SM_classical_gen_bound} yields
\begin{equation}
\BCg
=
\sum_{i=1}^n w_i\,\vect b^\top[\mathcal C^{(i)}(m)]^+\vect b
\simeq
\sum_{i=1}^n w_i\frac{b_i^2}{m[\mathcal C^{(i)}(1)]_{ii}}.
\end{equation}

Now consider the global Cram\'er--Rao bound. Since $[\mathcal C^{(i)}(m)]_{ii}=m[\mathcal C^{(i)}(1)]_{ii}$, we have
\begin{equation}
\BCd
=
\sum_{i=1}^n w_i\frac{b_i^2}{[\mathcal C^{(i)}(m)]_{ii}}
=
\sum_{i=1}^n w_i\frac{b_i^2}{m[\mathcal C^{(i)}(1)]_{ii}}.
\end{equation}
Comparing the two expressions confirms that $\BCg\simeq \BCd$ in the limit $m\to\infty$.

The maximum-likelihood estimator asymptotically achieves the Cram\'er--Rao bound in the regime $m\gg 1$. For the weighted variance with positive weights $w_i$, the MLE therefore saturates the global Cram\'er--Rao bound,
\begin{equation}
\overline{\var}(\hat\theta_{\rm MLE})
\approx
\sum_{i=1}^n w_i\frac{b_i^2}{m\FC(\theta_i)}
=
\BCd,
\end{equation}
where $\FC(\theta_i)$ denotes the classical Fisher information for a single measurement at $\theta_i$, corresponding to $[\mathcal C^{(i)}(1)]_{ii}$ in our notation.
Since the fully global bound reduces to the global Cram\'er--Rao bound in this limit, the MLE also saturates the fully global bound,
\begin{equation}
\overline{\var}(\hat\theta_{\rm MLE})\approx \BCg.
\end{equation}
This proves both the asymptotic equivalence of the two bounds and their saturation by the MLE.
\end{proof}

\section{Observation~2: Broad-Domain Attainability Condition}

\textbf{Observation~2.}---
Consider fixed finite repetition number \(m\), including non-asymptotic regimes in which the data do not localize the parameter. Let the sampled parameter values
\(\Theta=\{\theta_1,\ldots,\theta_n\}\) provide an increasingly dense sampling of a broad domain. Then the relevant attainability question is whether the global score family spans the estimator fluctuations required by the fully global saturation equations. When the compatibility/range condition derived below is satisfied, the fully global level is attainable,
\begin{equation}\label{eq:observation_2}
\overline{\var}(\hat\theta)\simeq\BCg .
\end{equation}
A useful sufficient condition is that the score matrix \(G=(\vect g(x))_{x\in X}\) has full column rank,
\begin{equation}\label{eq:obs2_full_rank_sufficient}
\operatorname{rank}(G)=|X|.
\end{equation}
Increasing the number of sampled parameter values can provide this rank condition when the corresponding score vectors become linearly independent.

\begin{proof}
By Eq.~\eqref{eq:SM_general_estimator_solution_condition}, a necessary and sufficient condition for the existence of an estimator saturating the fully global classical bound is
\begin{equation}\label{eq:SM_obs2_existence}
P_i\left(\vect h-\vect 1\,\vect p_i^\top\vect h\right)\in \operatorname{Range}(G^\top),
\qquad
\forall i:\; w_i\neq0,
\end{equation}
where \(\vect h=(\hat\theta(x))_{x\in X}\in\mathbb R^{|X|}\) and \(G=(\vect g(x))_{x\in X}\in\mathbb R^{n\times |X|}\).
This is the compatibility/range condition: the centered estimator fluctuation at each active parameter value, weighted by the probability vector at that value, must lie in the range of the score map \(G^\top\).

If \(\operatorname{rank}(G)=|X|\), then \(G\) has full column rank. Since \(\operatorname{rank}(G^\top)=\operatorname{rank}(G)=|X|\) and \(G^\top:\mathbb R^n\to\mathbb R^{|X|}\), it follows that
\begin{equation}
\operatorname{Range}(G^\top)=\mathbb R^{|X|}.
\end{equation}
Therefore, for any estimator \(\hat\theta(x)\) and for every \(i\) with \(w_i\neq0\), the vector
\begin{equation}
P_i\left(\vect h-\vect 1\,\vect p_i^\top\vect h\right)
\end{equation}
automatically belongs to \(\operatorname{Range}(G^\top)\). Hence the compatibility condition~\eqref{eq:SM_obs2_existence} is automatically satisfied, and the fully global classical bound is tight.

For finite \(m\), increasing the number of sampled parameter values enlarges the ambient score space. If the vectors \(\{\vect g(x)\}_{x\in X}\) become linearly independent, then Eq.~\eqref{eq:obs2_full_rank_sufficient} holds, and Eq.~\eqref{eq:observation_2} follows.
\end{proof}


\section{Two-Outcome Example}

\subsection{Classical Information Matrix for \texorpdfstring{$m$}{m} Measurements}

We now specialize to the case of two-outcome measurements, where each measurement yields $x_j=\pm1$ and the probability distribution is
\begin{equation}\label{SM_2_outcomes}
p(x_j|\theta)=\frac{1+x_j f(\theta)}{2},
\end{equation}
with $f(\theta)$ a smooth function of the parameter $\theta$. For $m$ independent measurements, the joint probability distribution factorizes as $p(\vect x|\theta)=\prod_{j=1}^m p(x_j|\theta)$, where $\vect x=(x_1,\ldots,x_m)$.

Define $s_j(x)=\partial_\theta\log p(x|\theta_j)$. Then
\begin{equation}
\partial_\theta p(\vect x|\theta_j)=p(\vect x|\theta_j)\sum_{r=1}^m s_j(x_r),
\end{equation}
and the classical information matrix takes the form
\begin{equation}\label{SM_eg_m}
[\mathcal C^{(i)}(m)]_{jk}
=
\sum_{x_1,\ldots,x_m=\pm1}
\frac{p(\vect x|\theta_j)p(\vect x|\theta_k)}{p(\vect x|\theta_i)}
\left(\sum_{r=1}^m s_j(x_r)\right)
\left(\sum_{t=1}^m s_k(x_t)\right).
\end{equation}
Let $m_+$ denote the number of outcomes with $x_j=+1$ among the $m$ measurements. A straightforward calculation gives
\begin{equation}
\sum_{r=1}^m s_j(x_r)
=
m_+s_j(+1)+(m-m_+)s_j(-1)
=
\frac{f'_j}{1-f_j^2}\bigl(2m_+-m-mf_j\bigr),
\end{equation}
where $f_i=f(\theta_i)$ and $f_i'=\partial_\theta f(\theta_i)$.
Define
\begin{equation}
F_{jk}^\pm=\frac{p(\pm1|\theta_j)p(\pm1|\theta_k)}{p(\pm1|\theta_i)}
=\frac{(1\pm f_j)(1\pm f_k)}{2(1\pm f_i)},
\qquad
A_j=\frac{f'_j}{1-f_j^2}.
\end{equation}
Then Eq.~\eqref{SM_eg_m} becomes
\begin{equation}\label{SM_eg_2}
[\mathcal C^{(i)}(m)]_{jk}
=
A_jA_k
\sum_{m_+=0}^m
\binom{m}{m_+}
(F_{jk}^+)^{m_+}(F_{jk}^-)^{m-m_+}
\bigl(2m_+-m-mf_j\bigr)\bigl(2m_+-m-mf_k\bigr).
\end{equation}
To evaluate this sum, introduce the generating function
\begin{equation}
G_{jk}(s)
=
\sum_{m_+=0}^m
\binom{m}{m_+}
(F_{jk}^+)^{m_+}(F_{jk}^-)^{m-m_+}e^{s(2m_+-m)}.
\end{equation}
Then Eq.~\eqref{SM_eg_2} can be written as
\begin{equation}\label{Sm_eg_4}
[\mathcal C^{(i)}(m)]_{jk}
=
A_jA_k
\left[
\frac{d^2}{ds^2}
-m(f_j+f_k)\frac{d}{ds}
+m^2f_jf_k
\right]
G_{jk}(s)\Big|_{s=0}.
\end{equation}

Since
\begin{equation}
G_{jk}(s)=e^{-sm}\bigl(F_{jk}^+e^{2s}+F_{jk}^-\bigr)^m,
\end{equation}
its derivatives at $s=0$ are
\begin{align}
G_{jk}(0)&=(F_{jk}^++F_{jk}^-)^m=X^m, \nonumber\\
G_{jk}'(0)&=m(F_{jk}^++F_{jk}^-)^{m-1}(F_{jk}^+-F_{jk}^-)=mX^{m-1}Y, \nonumber\\
G_{jk}''(0)&=m(F_{jk}^++F_{jk}^-)^{m-2}\bigl[(m-1)Y^2+X^2\bigr]
= mX^{m-2}\bigl[(m-1)Y^2+X^2\bigr],
\end{align}
where
\begin{equation}
X=F_{jk}^++F_{jk}^-,
\qquad
Y=F_{jk}^+-F_{jk}^-.
\end{equation}
Substituting these expressions into Eq.~\eqref{Sm_eg_4} gives
\begin{equation}\label{Sm_eg_6}
[\mathcal C^{(i)}(m)]_{jk}
=
mA_jA_kX^{m-2}\Big[(m-1)Y^2+X^2-m(f_j+f_k)XY+mf_jf_kX^2\Big].
\end{equation}
Using
\begin{equation}
f_j+f_k=f_iX+Y,
\qquad
f_jf_k=X-1+f_iY,
\end{equation}
the bracket in Eq.~\eqref{Sm_eg_6} simplifies to
\begin{equation}
\bigl[1+m(X-1)\bigr](X^2-Y^2).
\end{equation}
Hence
\begin{equation}
[\mathcal C^{(i)}(m)]_{jk}
=
mA_jA_kX^{m-2}\bigl[1+m(X-1)\bigr](X^2-Y^2).
\end{equation}

A direct calculation gives
\begin{align}
X
&=
F_{jk}^++F_{jk}^-
=
1+\frac{(f_j-f_i)(f_k-f_i)}{1-f_i^2},
\nonumber\\
X^2-Y^2
&=
4F_{jk}^+F_{jk}^-
=
\frac{(1-f_j^2)(1-f_k^2)}{1-f_i^2},
\nonumber\\
A_j&=\frac{f'_j}{1-f_j^2}.
\end{align}
Substituting these expressions, we obtain
\begin{equation}
[\mathcal C^{(i)}(m)]_{jk}
=
m\bigl(1+\Delta_{i;jk}\bigr)^{m-2}\bigl(1+m\Delta_{i;jk}\bigr)\frac{f'_jf'_k}{1-f_i^2},
\end{equation}
where
\begin{equation}
\Delta_{i;jk}=\frac{(f_j-f_i)(f_k-f_i)}{1-f_i^2}.
\end{equation}

For a single measurement, this reduces to
\begin{equation}\label{SM_single_shot_CIM}
[\mathcal C^{(i)}(1)]_{jk}
=
\frac{f'_jf'_k}{1-f_i^2}.
\end{equation}
Therefore, the $m$-shot information matrix takes the form
\begin{equation}\label{SM_m_shot_CIM}
[\mathcal C^{(i)}(m)]_{jk}
=
m\,h(m,\Delta_{i;jk})\,[\mathcal C^{(i)}(1)]_{jk},
\end{equation}
which is the result quoted in the main text, with $h(m,x)=(1+x)^{m-2}(1+mx)$.

\subsection{Comparison with the Maximum-Likelihood Estimator}
We now turn to the single-shot case. Writing $\vect f=(f_1,\ldots,f_n)^\top$ and $\vect f'=(f_1',\ldots,f_n')^\top$, with $f_i=f(\theta_i)$ and $f_i'=\partial_\theta f(\theta_i)$, Eq.~\eqref{SM_single_shot_CIM} gives the rank-one form
\begin{equation}\label{eq:SM_single_shot_rankone}
\mathcal C^{(i)}(1)=\frac{\vect f'\vect f'^\top}{1-f_i^2}.
\end{equation}
For the choice $g_k(x)=\partial_\theta p(x|\theta_k)$, the bias vectors become identical for all $i$, namely $\vect b^{(i)}=\vect b$.
Since $\vect b\in\operatorname{Range}(\mathcal C^{(i)})$, we can write $\vect b=\alpha \vect f'$ for some scalar $\alpha$.

The global Cram\'er--Rao bound then becomes
\begin{equation}\label{SM_eg_diag}
\BCd
=
\sum_{i=1}^n w_i\frac{(b_{i}^{(i)})^2}{[\mathcal C^{(i)}]_{ii}}
=
\alpha^2\sum_{i=1}^n w_i(1-f_i^2).
\end{equation}
Using the pseudoinverse of a rank-one matrix, we further obtain
\begin{equation}
(\mathcal C^{(i)})^+
=
(1-f_i^2)\frac{\vect f'\vect f'^\top}{(\vect f'^\top\vect f')^2},
\qquad
\mathcal C_w^+
=
\frac{1}{\sum_i w_i/(1-f_i^2)}\,
\frac{\vect f'\vect f'^\top}{(\vect f'^\top\vect f')^2},
\end{equation}
where $\mathcal C_w=\sum_i w_i\mathcal C^{(i)}$.
Hence
\begin{equation}\label{SM_eg_gen}
\BCc
=
\vect b_w^\top\mathcal C_w^+\vect b_w
=
\frac{\alpha^2}{\sum_i w_i/(1-f_i^2)},
\qquad
\BCg
=
\sum_i w_i\,\vect b^\top(\mathcal C^{(i)})^+\vect b
=
\alpha^2\sum_i w_i(1-f_i^2),
\end{equation}
where we used $\vect b_w=\sum_i w_i\vect b=\vect b$.
Therefore,
\begin{equation}
\BCg=\BCd\ge \BCc.
\end{equation}
This is the single-shot two-outcome relation quoted in the main text.

For a single-shot binary measurement with $x\in\{+1,-1\}$, the maximum-likelihood estimator takes the form
\begin{equation}\label{sm_MLE}
\hat\theta_{\rm MLE}(x)
=
\begin{cases}
\theta_+, & x=+1,\\
\theta_-, & x=-1,
\end{cases}
\end{equation}
where $\theta_+=\arg\max_\theta f(\theta)$ and $\theta_-=\arg\min_\theta f(\theta)$.
Equivalently,
\begin{equation}
\hat\theta_{\rm MLE}(x)=\frac{\theta_++\theta_-}{2}+\frac{\theta_+-\theta_-}{2}\,x.
\end{equation}
Its expectation value and second moment are
\begin{equation}
\mathbb E_\theta[\hat\theta_{\rm MLE}(x)]
=
\frac{\theta_++\theta_-}{2}+\frac{\theta_+-\theta_-}{2}f(\theta),
\qquad
\mathbb E_\theta[\hat\theta_{\rm MLE}(x)^2]
=
\frac{\theta_+^2+\theta_-^2}{2}+\frac{\theta_+^2-\theta_-^2}{2}f(\theta).
\end{equation}
Hence $b_i=\partial_\theta \mathbb E_{\theta_i}[\hat\theta_{\rm MLE}(x)]=\alpha f_i'$ with $\alpha=(\theta_+-\theta_-)/2$, confirming $\vect b=\alpha\vect f'$, and
\begin{equation}
\var_{\theta_i}[\hat\theta_{\rm MLE}(x)]
=
\alpha^2(1-f_i^2).
\end{equation}
Therefore,
\begin{equation}\label{sm_MLE_result}
\overline{\var}(\hat\theta_{\rm MLE})
=
\sum_{i=1}^n w_i\,\var_{\theta_i}[\hat\theta_{\rm MLE}(x)]
=
\alpha^2\sum_{i=1}^n w_i(1-f_i^2)
=
\BCg
=
\BCd.
\end{equation}
Thus, in the single-shot two-outcome case, the fully global and global Cram\'er--Rao bounds coincide, while the global Barankin bound is looser. Moreover, the fully global bound is saturated by the maximum-likelihood estimator.


\section{Derivation of Quantum Bounds}

We now extend the classical hierarchy to the quantum setting by optimizing the classical bounds over all quantum measurements while keeping the bias constraints fixed.

\subsection{From Classical to Quantum Bounds}
\label{SM_cl_to_quantum}

Consider a quantum statistical model $\rho(\theta)$ measured by a positive-operator-valued measure (POVM) $\{E_x\}$, so that the outcome probabilities are $p(x|\theta)=\tr[E_x\rho(\theta)]$. Choosing test functions of the form $g_k(x)=\tr[E_xG_k]$, where $G_k$ is a Hermitian operator associated with $\rho(\theta_k)$, the quantum bound for a fixed hierarchy matrix $A$ is obtained by minimizing the classical bound over all POVMs,
\begin{equation}
\BQ(A)=\min_{\{E_x\}}\BC(A)
=
\frac{\left(\sum_{i=1}^n w_i\,\vect b^{(i)\top}\vect a_i\right)^2}
{\sum_{i=1}^n w_i\,\max_{\{E_x\}}\!\left[\vect a_i^{\top}\mathcal C^{(i)}\vect a_i\right]}.
\end{equation}
Thus it remains to determine $\max_{\{E_x\}} \vect a_i^{\top}\mathcal C^{(i)}\vect a_i$.

Using $p(x|\theta_i)=\tr[E_x\rho_i]$ with $\rho_i=\rho(\theta_i)$ and $g_j(x)=\tr[E_xG_j]$, we obtain
\begin{equation}\label{SM_quantumm_1}
\vect a_i^{\top}\mathcal C^{(i)}\vect a_{i}
=
\sum_{x\in X} \frac{\tr[E_x G_{A;i}]^2}{\tr[E_x\rho_i]},
\end{equation}
where
\begin{equation}
 G_{A;i}=\sum_{j=1}^n A_{ji}G_j.
\end{equation}

Following Ref.~\cite{PhysRevLett.130.260801}, we introduce the linear superoperator $\Omega_\rho$. For an operator $\rho=\sum_i p_i\ket{i}\bra{i}$ and any operator $X$,
\begin{equation}\label{SM_def_SLD}
\Omega_\rho(X)=\sum_{i,j}\frac{2}{p_i+p_j}\ket{i}\bra{i}X\ket{j}\bra{j},
\end{equation}
where terms with $p_i+p_j=0$ are omitted. In particular, when $X=\partial_\theta\rho(\theta)$, the operator $\Omega_{\rho(\theta)}(\partial_\theta\rho(\theta))$ is the standard symmetric logarithmic derivative (SLD). More generally, $\Omega_\rho(X)$ satisfies
\begin{equation}\label{SM_quantumm_SLD}
\frac{1}{2}\bigl[\Omega_\rho(X)\rho+\rho\,\Omega_\rho(X)\bigr]
=
X-\Pi_\perp X\Pi_\perp,
\end{equation}
where $\Pi_\perp$ projects onto the kernel of $\rho$.

Using Eq.~\eqref{SM_quantumm_SLD}, Eq.~\eqref{SM_quantumm_1} can be written as
\begin{equation}
\vect a_i^{\top}\mathcal C^{(i)}\vect a_{i}
=
\sum_{x\in X}
\frac{
\left(
\tr\!\left[E_x\,\tfrac{1}{2}\bigl(\Omega_{\rho_i}( G_{A;i})\rho_i+\rho_i\Omega_{\rho_i}( G_{A;i})\bigr)\right]
\right)^2
}
{\tr[E_x\rho_i]}.
\end{equation}
If necessary, one may regularize $\rho_i$ as $\rho_i+\epsilon I$ with $\epsilon\ll1$ and take the limit $\epsilon\to0$ at the end, so that $\Pi_\perp=0$.

We now use the Cauchy--Schwarz inequality $|\tr[A^\dagger B]|^2\le \tr[A^\dagger A]\tr[B^\dagger B]$ with $A=\sqrt{E_x}\sqrt{\rho_i}$ and $B=\sqrt{E_x}\Omega_{\rho_i}( G_{A;i})\sqrt{\rho_i}$. This gives
\begin{align}\label{SM_quantum_saturation}
\vect a_{i}^{\top}\mathcal C^{(i)}\vect a_i
&=
\sum_{x\in X} \frac{\bigl(\Re\,\tr[E_x\Omega_{\rho_i}( G_{A;i})\rho_i]\bigr)^2}{\tr[E_x\rho_i]}
\nonumber\\
&\le
\sum_{x\in X} \frac{|\tr[E_x\Omega_{\rho_i}( G_{A;i})\rho_i]|^2}{\tr[E_x\rho_i]}
\nonumber\\
&\le
\sum_{x\in X} \tr[E_x\Omega_{\rho_i}( G_{A;i})\rho_i\Omega_{\rho_i}( G_{A;i})]
\nonumber\\
&=
\tr[ G_{A;i}\Omega_{\rho_i}( G_{A;i})]\n
&=
\sum_{j,k=1}^n A_{ji}\,\tr[G_j\Omega_{\rho_i}(G_k)]\,A_{ki}.
\end{align}
This motivates the definition of the quantum information matrix
\begin{equation}\label{SM_QIM}
[\mathcal Q^{(i)}]_{jk}=\tr[G_j\Omega_{\rho_i}(G_k)].
\end{equation}
Hence
\begin{equation}
\vect a_i^{\top}\mathcal C^{(i)}\vect a_{i}
\le
\vect a_i^{\top}\mathcal Q^{(i)}\vect a_{i},
\end{equation}
and the quantum bound for a fixed hierarchy matrix $A$ is
\begin{equation}\label{SM_BQA}
\BQ(A)
=
\frac{\left(\sum_{i=1}^n w_i\,\vect a_{i}^\top\vect b^{(i)}\right)^2}
{\sum_{i=1}^n w_i\,\vect a_i^{\top}\mathcal Q^{(i)}\vect a_{i}}.
\end{equation}

Applying the Rayleigh quotient maximization theorem exactly as in Sec.~\ref{SM_Optimizatio_A}, but with $\mathcal C^{(i)}$ replaced by $\mathcal Q^{(i)}$, we obtain the quantum hierarchy
\begin{align}\label{SM_Quantum_bounds}
\BQd&=\min_{\{E_x\}}\BCd=\sum_{i=1}^n w_i\frac{(b_{i}^{(i)})^2}{[\mathcal Q^{(i)}]_{ii}},\n
\BQc&=\min_{\{E_x\}}\BCc=\vect b_w^\top\mathcal Q_w^+\vect b_w,\n
\BQg&=\min_{\{E_x\}}\BCg=\sum_{i=1}^n w_i\,\vect b^{(i)\top}(\mathcal Q^{(i)})^+\vect b^{(i)},
\end{align}
which correspond to the global Cram\'er--Rao, global Barankin, and fully global quantum bounds in the main text.

\subsection{Automatic Finite Condition}
\label{SM_quantum_finite_cond}

Since the derivation of the quantum bounds also uses the Rayleigh quotient maximization theorem, we must verify the finite condition
\begin{equation}\label{eq:SM_quantum_finite_condition}
\vect b^{(i)}\in \operatorname{Range}(\mathcal Q^{(i)}),
\end{equation}
which ensures that the quantum bounds remain finite.

Let $\vect a\in\ker(\mathcal Q^{(i)})$, so that $\mathcal Q^{(i)}\vect a=0$. Then $\vect a^\top\mathcal Q^{(i)}\vect a=0$. Since $\mathcal Q^{(i)}\ge \mathcal C^{(i)}$ by construction and $\mathcal C^{(i)}$ is positive semidefinite, we obtain
\begin{equation}
0=\vect a^\top\mathcal Q^{(i)}\vect a \ge \vect a^\top\mathcal C^{(i)}\vect a\ge 0,
\end{equation}
which implies $\vect a^\top\mathcal C^{(i)}\vect a=0$.

Using the definition of the classical information matrix,
\begin{equation}
\mathcal C^{(i)}=\sum_{x\in X} \frac{\vect g(x)\vect g(x)^\top}{p(x|\theta_i)},
\end{equation}
we have
\begin{equation}
\vect a^\top\mathcal C^{(i)}\vect a
=
\sum_{x\in X} \frac{[\vect a^\top\vect g(x)]^2}{p(x|\theta_i)}
=
0.
\end{equation}
Since each term is nonnegative, this implies
\begin{equation}
\vect a^\top\vect g(x)=0,\qquad \forall x\in X.
\end{equation}

Now recall that
\begin{equation}
\vect b^{(i)}=\sum_{x\in X} \vect g(x)\,\Delta\hat\theta_i(x).
\end{equation}
Taking the inner product with $\vect a$ yields
\begin{equation}
\vect a^\top \vect b^{(i)}
=
\sum_{x\in X} [\vect a^\top\vect g(x)]\,\Delta\hat\theta_i(x)
=
0.
\end{equation}
Hence every vector in $\ker(\mathcal Q^{(i)})$ is orthogonal to $\vect b^{(i)}$, i.e. $\vect b^{(i)}\perp\ker(\mathcal Q^{(i)})$.
For a real symmetric matrix, $\operatorname{Range}(\mathcal Q^{(i)})=(\ker \mathcal Q^{(i)})^\perp$, and therefore
\begin{equation}
\vect b^{(i)}\in\operatorname{Range}(\mathcal Q^{(i)}).
\end{equation}
Thus the finite condition holds automatically.

\subsection{Saturation Condition}

We now collect the saturation conditions underlying the fully global quantum bound.
In Eq.~\eqref{SM_quantum_saturation}, the derivation of the quantum bound for a fixed hierarchy matrix used the passage from the real part to the modulus and then a Cauchy--Schwarz inequality for each outcome. Equality requires, for every \(x\) and every active \(i\),
\begin{align}
&\text{(i)}\quad {\rm Im}\,\tr\!\left[E_x\Omega_{\rho_i}( G_{A;i})\rho_i\right]=0,\\
&\text{(ii)}\quad \sqrt{E_x}\sqrt{\rho_i}
=
\alpha_{x i}\,\sqrt{E_x}\Omega_{\rho_i}( G_{A;i})\sqrt{\rho_i},
\end{align}
with a scalar \(\alpha_{x i}\in\mathbb C\) that may depend on the outcome and on the active parameter value.
To obtain the fully global quantum bound, we must also impose the optimal hierarchy-matrix condition
\begin{equation}
\text{(iii)}\quad
\vect a_{i}=(\mathcal Q^{(i)})^+\vect b^{(i)},
\qquad
\forall i:\; w_i\neq 0.
\end{equation}

Condition (ii) may be written with real \(\alpha_{x i}\) whenever condition (i) holds. Multiplying condition (ii) on the left by \(\sqrt{\rho_i}\sqrt{E_x}\) and taking the trace gives
\begin{equation}
\tr[E_x\rho_i]
=
\alpha_{x i}\,\tr[E_x\Omega_{\rho_i}( G_{A;i})\rho_i].
\end{equation}
Thus the saturation conditions for \(\BQg=\BCg\) can be written as follows: for all outcomes \(x\) and all active \(i\), there exist real scalars \(\alpha_{x i}\) such that
\begin{align}
&\text{(a)}\quad
\sqrt{E_x}\sqrt{\rho_i}
=
\alpha_{x i}\,\sqrt{E_x}\Omega_{\rho_i}( G_{A;i})\sqrt{\rho_i},
\label{SM_a}\\
&\text{(b)}\quad
\vect a_{i}=(\mathcal Q^{(i)})^+\vect b^{(i)}.
\label{SM_b}
\end{align}
where
\begin{equation}\label{SM_generalized_SLD}
 G_{A;i}=\sum_{j=1}^n A_{ji}G_j.
\end{equation}

Condition (a) ensures that the measurement is optimal for a given hierarchy matrix $A$, while condition (b) guarantees that the hierarchy matrix itself is optimally chosen. Together, these characterize the optimal measurement and the optimal hierarchy matrix that achieve the fully global quantum bound. The corresponding conditions for the other quantum bounds are analogous and will not be discussed separately.

\section{Optimal Quantum Measurements}

\subsection{Proof of Observation~3}

We now discuss the existence and construction of optimal measurements for the fully global quantum bound. 
Throughout this section, we assume the regularity condition \(p(x|\theta)>0\) for all \(\theta\in\Theta\) and \(x\in X\).

\textbf{Observation~3.}
Let $G_j=\partial_\theta\rho(\theta_j)$ and $\mathcal Y=\operatorname{Span}\{G_j\mid j=1,\dots,n\}$.
Suppose that the generalized SLD maps are closed on the global score span up to scalar operators,
\begin{equation}\label{eq:SM_identity_extended_condition}
\Omega_{\rho_i}(\mathcal Y)\subseteq \mathcal Y\oplus \mathbb R I,
\qquad \forall i:w_i\neq0 .
\end{equation}
Then there exists an optimal measurement \(\{E_x^{\rm opt}\}\) such that $\BQg=\BCg(\{E_x^{\rm opt}\}).$
The POVM elements \(E_x^{\rm opt}\) are the eigenprojectors of
\begin{equation}\label{eq:SM_obs_3}
M_{\rm FG}\equiv
\sum_{j=1}^n (T^+\vect b)_jG_j
=
\sum_x p_xE_x^{\rm opt},
\end{equation}
where $p_x\in\mathbb R$ are the corresponding eigenvalues, 
$\vect b$ is the bias vector defined in Eq.~\eqref{eq:SM_bias_matrix} with $\vect b^{(i)}=\vect b$ for all $i$, and
\begin{equation}
T_{jk}=\tr[G_jG_k].
\end{equation}

\begin{proof}
Our aim is to construct POVMs satisfying the saturation condition~\eqref{SM_a} together with the optimal hierarchy-matrix condition~\eqref{SM_b} for the fully global quantum bound.

Since we consider the test observables $G_j=\partial_\theta\rho(\theta_j)$, the bias vector is independent of $i$, i.e.,
\begin{equation}
[\vect b^{(i)}]_j
=
\sum_x \partial_\theta p(x|\theta_j)
\bigl(\hat\theta(x)-\mathbb E_{\theta_i}[\hat\theta(x)]\bigr)
=
\partial_\theta\sum_x p(x|\theta_j)\hat\theta(x).
\end{equation}
Hence we may write $\vect b^{(i)}=\vect b$ for all $i$.
For convenience, we use the vectorization map $O\mapsto |O\rangle\rangle$, defined by $|O\rangle\rangle=\sum_{m,n}O_{mn}|m\rangle\otimes|n\rangle$, with inner product $\langle\langle P|O\rangle\rangle=\tr[P^\dagger O]$. Let $|\partial_i\rho\rangle\rangle\equiv|\partial_\theta\rho(\theta_i)\rangle\rangle$ and
\begin{equation}
\mathcal Y=\operatorname{Span}\{|\partial_i\rho\rangle\rangle\mid i=1,\dots,n\}.
\end{equation}
By the identity-extended invariance condition~\eqref{eq:SM_identity_extended_condition}, the action of $\Omega_{\rho_i}$ on the score span can be written as
\begin{equation}\label{eq:SM_obs3_1_modified}
\Omega_{\rho_i}|\partial_j\rho\rangle\rangle
=
\sum_{k=1}^n V_{kj}^{(i)}|\partial_k\rho\rangle\rangle
+
r_{ij}^{(i)}|I\rangle\rangle,
\end{equation}
where $V^{(i)}$ is a real matrix and $r_{ij}^{(i)}\in\mathbb R$. Thus, the only component generated outside the global score span is the identity direction.
Using Eq.~\eqref{SM_generalized_SLD}, the generalized SLD associated with the $i$-th column of the hierarchy matrix satisfies
\begin{align}\label{eq:SM_obs3_key_modified}
\Omega_{\rho_i}| G_{A;i}\rangle\rangle
&=
\sum_j A_{ji}\Omega_{\rho_i}|\partial_j\rho\rangle\rangle =
\sum_{j,k}A_{ji}V_{kj}^{(i)}|\partial_k\rho\rangle\rangle
+
\sum_j A_{ji}r_{ij}^{(i)}|I\rangle\rangle =
\sum_k [V^{(i)}\vect a^{(i)}]_k|\partial_k\rho\rangle\rangle
+
r_i(A)|I\rangle\rangle,
\end{align}
where $\vect a^{(i)}=(A_{1i},\dots,A_{ni})^\top$ and $r_i(A)=\sum_j A_{ji}r_{ij}^{(i)}$.

On the other hand, since $G_j=\partial_\theta\rho(\theta_j)$ is traceless, $\tr(G_jI)=\partial_\theta\tr\rho(\theta_j)=0$. Therefore, the identity component in Eq.~\eqref{eq:SM_obs3_1_modified} does not contribute to the quantum information matrix:
\begin{align}\label{eq:SM_obs3_2_modified}
[\mathcal Q^{(i)}]_{jk}
&=
\langle\langle\partial_j\rho|\Omega_{\rho_i}|\partial_k\rho\rangle\rangle =
\sum_{m=1}^n V_{mk}^{(i)}
\langle\langle\partial_j\rho|\partial_m\rho\rangle\rangle =
\sum_{m=1}^n T_{jm}V_{mk}^{(i)} .
\end{align}
With $T_{jk}=\langle\langle\partial_j\rho|\partial_k\rho\rangle\rangle=\tr[(\partial_\theta\rho(\theta_j))(\partial_\theta\rho(\theta_k))]$, this gives
\begin{equation}\label{eq:SM_obs3_00_modified}
\mathcal Q^{(i)}=TV^{(i)}.
\end{equation}
From the finite condition proved in Sec.~\ref{SM_quantum_finite_cond}, we know that $\vect b\in\operatorname{Range}(\mathcal Q^{(i)})$. The optimal hierarchy-matrix condition~\eqref{SM_b} therefore gives $\mathcal Q^{(i)}\vect a_{i}=\vect b$. Using Eq.~\eqref{eq:SM_obs3_00_modified}, this becomes $TV^{(i)}\vect a_{i}=\vect b$. Hence,
\begin{equation}\label{eq:SM_obs3_imo0_modified}
V^{(i)}\vect a_{i}
=
T^+\vect b+\vect t^{(i)},
\qquad
\vect t^{(i)}\in\ker(T).
\end{equation}
Substituting Eq.~\eqref{eq:SM_obs3_imo0_modified} into Eq.~\eqref{eq:SM_obs3_key_modified}, we obtain
\begin{align}\label{eq:SM_obs3_imp_2_modified}
\Omega_{\rho_i}| G_{A;i}\rangle\rangle
&=
\sum_k [T^+\vect b]_k|\partial_k\rho\rangle\rangle
+
\sum_k t^{(i)}_{k}|\partial_k\rho\rangle\rangle
+
r_i(A)|I\rangle\rangle.
\end{align}
The second term in Eq.~\eqref{eq:SM_obs3_imp_2_modified} vanishes identically. Indeed, since $\vect t^{(i)}\in\ker(T)$,
\begin{equation}
\left\|
\sum_k t^{(i)}_{k}G_k
\right\|_{\rm HS}^2
=
(\vect t^{(i)})^\top T\vect t^{(i)}
=
0.
\end{equation}
Because $\sum_k t^{(i)}_{k}G_k$ is Hermitian, this implies $\sum_k t^{(i)}_{k}G_k=0$. Therefore Eq.~\eqref{eq:SM_obs3_imp_2_modified} reduces to
\begin{equation}\label{eq:SM_optimized_SLD_scalar_shift}
\Omega_{\rho_i}( G_{A;i})
=
M_{\rm FG}+r_i(A)I,
\qquad
M_{\rm FG}=\sum_k(T^+\vect b)_kG_k.
\end{equation}

Let
$M_{\rm FG}=\sum_xp_xE_x^{\rm opt}$
be its spectral decomposition. Then
\[
\Omega_{\rho_i}( G_{A;i})
=
\sum_x[p_x+r_i(A)]E_x^{\rm opt}.
\]
Therefore
\[
E_x^{\rm opt}\Omega_{\rho_i}( G_{A;i})
=
[p_x+r_i(A)]E_x^{\rm opt}.
\]
Equivalently,
\[
\sqrt{E_x^{\rm opt}}\Omega_{\rho_i}( G_{A;i})\sqrt{\rho_i}
=
[p_x+r_i(A)]\sqrt{E_x^{\rm opt}}\sqrt{\rho_i}.
\]
Hence Eq.~\eqref{SM_a} is satisfied with $\alpha_{x,i}=\frac{1}{p_x+r_i(A)}$
for every nonzero eigenvalue \(p_x+r_i(A)\). Outcomes with $p_x+r_i(A)=0$ have zero score contribution and saturate the Cauchy--Schwarz inequality trivially.
Thus both Eqs.~\eqref{SM_a} and~\eqref{SM_b} are satisfied. Consequently,
$\BQg=\BCg(\{E_x^{\rm opt}\})$,
with $E_x^{\rm opt}$ given by the eigenprojectors of $M_{\rm FG}$.
\end{proof}

\subsection{Example}

We illustrate Observation~3 with a qubit model
\begin{equation}\label{eq:SM_example}
\rho(\theta)=\frac{\mathbb I+\vect r_\theta^\top\vect\sigma}{2},
\qquad
\vect r_\theta=(d\cos\theta,d\sin\theta,0),
\end{equation}
where $\vect\sigma=(\sigma_x,\sigma_y,\sigma_z)$ and $0<d<1$.

We take the parameter domain $\Theta=\{\theta_1=0,\theta_2=\pi/4\}$ and choose the test observables $G_k=\partial_\theta\rho(\theta_k)$. Let $\rho_i=\rho(\theta_i)$ and denote its eigenstates by $\ket{\pm_i}$, with eigenvalues $p_{i;\pm}=(1\pm d)/2$. Using Eq.~\eqref{SM_def_SLD}, we obtain
\begin{align}
\tr[(\partial_\theta\rho_j)\Omega_{\rho_i}(\partial_\theta\rho_k)]
&=
\sum_{\alpha,\beta=\pm}
\frac{2}{p_{i;\alpha}+p_{i;\beta}}
\bra{\alpha_i}\partial_\theta\rho_j\ket{\beta_i}
\bra{\beta_i}\partial_\theta\rho_k\ket{\alpha_i}
\nonumber\\
&=
\frac{2}{1+d}\bra{+_i}\partial_\theta\rho_j\ket{+_i}\bra{+_i}\partial_\theta\rho_k\ket{+_i}
+
\frac{2}{1-d}\bra{-_i}\partial_\theta\rho_j\ket{-_i}\bra{-_i}\partial_\theta\rho_k\ket{-_i}
\nonumber\\
&\quad
+2\bra{+_i}\partial_\theta\rho_j\ket{-_i}\bra{-_i}\partial_\theta\rho_k\ket{+_i}
+
2\bra{-_i}\partial_\theta\rho_j\ket{+_i}\bra{+_i}\partial_\theta\rho_k\ket{-_i}
\nonumber\\
&=
\frac{d^4\sin(\theta_{ij})\sin(\theta_{ik})}{1-d^2}
+
d^2\cos(\theta_{jk}),
\end{align}
where $\theta_{ij}\equiv \theta_i-\theta_j$, and we used
\begin{equation}
\bra{\pm_i}\partial_\theta\rho_j\ket{\pm_i}
=
\pm\frac{d\sin(\theta_{ij})}{2},
\qquad
\bra{+_i}\partial_\theta\rho_j\ket{-_i}
=
\frac{id\cos(\theta_{ij})}{2}.
\end{equation}
For $\Theta=\{0,\pi/4\}$, the quantum information matrices defined in Eq.~\eqref{SM_QIM} are
\begin{equation}\label{sm_QQ}
\mathcal Q^{(1)}
=
d^2
\begin{pmatrix}
1 & 1/\sqrt{2}\\
1/\sqrt{2} & (2-d^2)/(2-2d^2)
\end{pmatrix},
\qquad
\mathcal Q^{(2)}
=
d^2
\begin{pmatrix}
(2-d^2)/(2-2d^2) & 1/\sqrt{2}\\
1/\sqrt{2} & 1
\end{pmatrix}.
\end{equation}

As an example, we choose the bias vector
\begin{equation}
\vect b=(1,1/\sqrt{2})^\top,
\end{equation}
which satisfies the finite condition $\vect b\in\operatorname{Range}(\mathcal Q^{(i)})$. Since $\mathcal Q^{(i)}$ is invertible for $0<d<1$, the quantum bounds in Eq.~\eqref{SM_Quantum_bounds} evaluate to
\begin{align}\label{Sm_eg_g}
&\BQd
=
\sum_{i=1}^n w_i\frac{B_{ii}^2}{[\mathcal Q^{(i)}]_{ii}}
=
\frac{w_1}{d^2}
+
\frac{w_2}{2d^2},\n 
&\BQg
=
\sum_{i=1}^n w_i\,\vect b^\top(\mathcal Q^{(i)})^+\vect b
=
\frac{w_1}{d^2}
+
\frac{w_2(2-d^2)}{2d^2},\n
&\mathcal B_{\rm Q}^{\rm GBar}
=
\vect b^\top\mathcal Q_w^{-1}\vect b
=
\frac{
(1-d^2)(2-d^2w_2)
}{
d^2\left[2(1-d^2)+d^4w_1w_2\right]
},
\end{align}
where $B_{11}=1$ and $B_{22}=1/\sqrt{2}$. Hence, $\BQd<\BQg$ and $\BQc<\BQg$ whenever $w_2\neq 0$ and $d\neq 1$.

We now show that the smaller global quantum Cram\'er--Rao bound cannot be attained by the optimal measurement that saturates the fully global quantum bound. The correlation matrix appearing in Eq.~\eqref{eq:SM_obs_3} is
\begin{equation}
T
=
d^2
\begin{pmatrix}
1 & 1/\sqrt{2}\\
1/\sqrt{2} & 1
\end{pmatrix},
\end{equation}
which gives
\begin{equation}
T^{-1}\vect b=(1/d^2,0)^\top.
\end{equation}
Therefore,
\begin{equation}
\sum_{j=1}^n (T^+\vect b)_j G_j
=
\frac{1}{d^2}\,\partial_\theta\rho_1
=
\frac{1}{2d}\sigma_y.
\end{equation}
According to Observation~3, the optimal quantum measurement is thus the projective measurement along the $y$ direction,
\begin{equation}
E_x^{\rm opt}=\frac{\mathbb I+x\sigma_y}{2},
\qquad x=\pm1.
\end{equation}

For this measurement, the outcome probabilities are
\begin{equation}
p(x|\theta)=\frac{1+xd\sin\theta}{2},
\end{equation}
which is precisely the two-outcome form of Eq.~\eqref{SM_2_outcomes} with $f(\theta)=d\sin\theta$.
Using the classical results in Eqs.~\eqref{SM_eg_diag} and \eqref{SM_eg_gen}, we obtain
\begin{align}\label{Sm_eg_c}
\BCg(\{E_x^{\rm opt}\})
&=
\BCd(\{E_x^{\rm opt}\})
=
\alpha^2\sum_i w_i(1-f_i^2)
=
\frac{w_1}{d^2}
+
\frac{w_2(2-d^2)}{2d^2},
\nonumber\\
\BCc(\{E_x^{\rm opt}\})
&=
\frac{\alpha^2}{\sum_i w_i/(1-f_i^2)}
=
\frac{2-d^2}{
d^2\left[w_1(2-d^2)+2w_2\right]
},
\end{align}
where $\alpha=1/d$ follows from $\vect b=\alpha \vect f'$.

Comparing Eqs.~\eqref{Sm_eg_g} and \eqref{Sm_eg_c}, we find
\begin{equation}
\BCg(\{E_x^{\rm opt}\})
=
\BQg .
\end{equation}
Thus the fully global quantum bound is saturated by the optimal measurement.
By contrast, the lower levels of the quantum hierarchy are not saturated by this measurement. 
Indeed, for $0<d<1$ and $w_1,w_2>0$,
\begin{equation}
\BCd(\{E_x^{\rm opt}\})
>
\BQd,
\qquad
\BCc(\{E_x^{\rm opt}\})
>
\BQc .
\end{equation}
Therefore, in this example, only the fully global quantum bound is tight, whereas the global quantum Cram\'er--Rao and global quantum Barankin bounds are strictly lower than the corresponding classical bounds generated by the optimal measurement.

\section{Generalized Van Trees Bound in Bayesian Estimation}

\subsection{Derivation of Classical Bounds}

We now extend the same organizing principle to Bayesian estimation. Consider a single continuous parameter $\theta\in\mathbb R$ with prior distribution $p(\theta)$, normalized as $\int d\theta\,p(\theta)=1$. The relevant figure of merit is the Bayesian mean-square error
\begin{equation}
\mathrm{BMSE}(\hat\theta)
=
\int d\theta\,p(\theta)\,\mathbb E_\theta\!\left[(\hat\theta(x)-\theta)^2\right],
\end{equation}
where $\mathbb E_\theta[\cdot]=\sum_x(\cdot)\,p(x|\theta)$.

In analogy with the discrete hierarchy matrix, we introduce a continuous hierarchy function $A(\theta,\theta')$ and define the Bayesian score function
\begin{equation}\label{eq:SM_bayesian_score}
S_{\rm B}(x,\theta)
=
\frac{\int d\theta'\,A(\theta,\theta')\,\partial_{\theta'}p(x,\theta')}
{p(x,\theta)},
\end{equation}
where $p(x,\theta)=p(x|\theta)p(\theta)$ is the joint probability density.
Applying the Cauchy--Schwarz inequality to $\hat\theta(x)-\theta$ and $S_{\rm B}(x,\theta)$ with respect to the joint measure $p(x,\theta)$ gives
\begin{equation}\label{eq:SM_CS_bayesian}
\left(
\int d\theta\,p(\theta)\,\mathbb E_\theta[(\hat\theta(x)-\theta)S_{\rm B}(x,\theta)]
\right)^2
\le
\mathrm{BMSE}(\hat\theta)
\int d\theta\,p(\theta)\,\mathbb E_\theta[S_{\rm B}(x,\theta)^2].
\end{equation}

We first evaluate the left-hand side. Using Eq.~\eqref{eq:SM_bayesian_score},
\begin{equation}\label{SM_B1}
\int d\theta\,p(\theta)\,\mathbb E_\theta[(\hat\theta(x)-\theta)S_{\rm B}(x,\theta)]
=
\int d\theta \sum_x (\hat\theta(x)-\theta)\int d\theta'\,A(\theta,\theta')\,\partial_{\theta'}p(x,\theta').
\end{equation}
Assume that the hierarchy function satisfies the regularity conditions
\begin{equation}\label{SM_coni_ii}
\lim_{\theta'\to\pm\infty}A(\theta,\theta')\,p(x,\theta')=0,
\qquad
\partial_{\theta'}A(\theta,\theta')=-\partial_\theta A(\theta,\theta'),
\end{equation}
together with
\begin{equation}\label{SM_coniii_iv}
\lim_{\theta\to\pm\infty}A(\theta,\theta')(\hat\theta(x)-\theta)=0,
\qquad
\int d\theta\,A(\theta,\theta')=1.
\end{equation}
Integrating by parts first over $\theta'$ and then over $\theta$, we obtain
\begin{align}
\int d\theta\,p(\theta)\,\mathbb E_\theta[(\hat\theta(x)-\theta)S_{\rm B}(x,\theta)]
&=
-\int d\theta \sum_x (\hat\theta(x)-\theta)\int d\theta'\,p(x,\theta')\,\partial_{\theta'}A(\theta,\theta')
\nonumber\\
&=
\int d\theta \sum_x (\hat\theta(x)-\theta)\int d\theta'\,p(x,\theta')\,\partial_\theta A(\theta,\theta')
\nonumber\\
&=
-\int d\theta' \sum_x p(x,\theta')\int d\theta\,A(\theta,\theta')\,\partial_\theta(\hat\theta(x)-\theta)
\nonumber\\
&=
\int d\theta' \sum_x p(x,\theta')\int d\theta\,A(\theta,\theta')
=
1.
\end{align}

The second factor on the right-hand side of Eq.~\eqref{eq:SM_CS_bayesian} becomes
\begin{align}
\int d\theta\,p(\theta)\,\mathbb E_\theta[S_{\rm B}(x,\theta)^2]
&=
\int d\theta \sum_x
\frac{\left[\int d\theta'\,A(\theta,\theta')\,\partial_{\theta'}p(x,\theta')\right]^2}
{p(x,\theta)}
\nonumber\\
&=
\iiint d\theta\,d\theta'\,d\theta''\,
A(\theta,\theta')\,\mathcal C_{\rm B}^{(\theta)}(\theta',\theta'')\,A(\theta,\theta''),
\end{align}
where
\begin{equation}\label{eq:SM_bayesian_info_kernel}
\mathcal C_{\rm B}^{(\theta)}(\theta',\theta'')
=
\sum_x
\frac{[\partial_{\theta'}p(x,\theta')][\partial_{\theta''}p(x,\theta'')]}
{p(x,\theta)}
\end{equation}
is the Bayesian information kernel.

Substituting these expressions into Eq.~\eqref{eq:SM_CS_bayesian} yields the classical Bayesian bound
\begin{equation}\label{eq:SM_bayesian_bound}
\mathrm{BMSE}(\hat\theta)
\ge
\frac{1}{
\iiint d\theta\,d\theta'\,d\theta''\,
A(\theta,\theta')\,\mathcal C_{\rm B}^{(\theta)}(\theta',\theta'')\,A(\theta,\theta'')
}.
\end{equation}
This bound holds for any hierarchy function $A(\theta,\theta')$ satisfying the regularity conditions in Eqs.~\eqref{SM_coni_ii} and \eqref{SM_coniii_iv}.

\subsection{Choice of Hierarchy Function}

A natural choice for the hierarchy function is
\begin{equation}\label{eq:SM_kernel}
A_\epsilon(\theta,\theta')
=
\frac{1}{\epsilon}K\!\left(\frac{\theta-\theta'}{\epsilon}\right),
\end{equation}
where the kernel $K$ satisfies
\begin{equation}\label{SM_cond_K}
\int du\,K(u)=1,
\qquad
\lim_{|u|\to\infty}uK(u)=0,
\end{equation}
and $\epsilon>0$ is a bandwidth parameter.
It is straightforward to verify that Eq.~\eqref{eq:SM_kernel} satisfies the regularity conditions in Eqs.~\eqref{SM_coni_ii} and \eqref{SM_coniii_iv}.

\subsection{Van Trees Bound as a Special Limit}

In the limit $\epsilon\to0$, the hierarchy function approaches a Dirac delta function,
\begin{equation}
\lim_{\epsilon\to0}A_\epsilon(\theta,\theta')=\delta(\theta-\theta').
\end{equation}
Equation~\eqref{eq:SM_bayesian_bound} then reduces to the Van Trees bound,
\begin{equation}
\mathrm{BMSE}(\hat\theta)
\ge
\frac{1}{\int d\theta\,\mathcal C_{\rm B}^{(\theta)}(\theta,\theta)}.
\end{equation}
Indeed,
\begin{align}
\int d\theta\,\mathcal C_{\rm B}^{(\theta)}(\theta,\theta)
&=
\sum_x \int d\theta\,
\frac{[p(x|\theta)\partial_\theta p(\theta)+p(\theta)\partial_\theta p(x|\theta)]^2}
{p(x|\theta)p(\theta)}
\nonumber\\
&=
\int d\theta\,p(\theta)\left(\sum_x \frac{[\partial_\theta p(x|\theta)]^2}{p(x|\theta)}\right)
+
\int d\theta\,\frac{[\partial_\theta p(\theta)]^2}{p(\theta)}.
\end{align}
The first term is the prior-weighted classical Fisher information, while the second is the contribution from the prior. Hence, the Van Trees bound appears as the local limit of the more general hierarchy-based Bayesian construction.

\subsection{From Classical to Quantum Bayesian Bounds}

Following the same strategy as in Sec.~\ref{SM_cl_to_quantum}, the classical Bayesian bound admits a direct quantum generalization. Replacing the classical joint distribution by the generalized quantum state
\begin{equation}
\rho[\theta]=p(\theta)\rho(\theta),
\end{equation}
we obtain the quantum Bayesian bound
\begin{equation}\label{eq:SM_bayesian_bound_qm}
\mathrm{BMSE}(\hat\theta)
\ge
\frac{1}{
\iiint d\theta\,d\theta'\,d\theta''\,
A(\theta,\theta')\,\mathcal Q_{\rm B}^{(\theta)}(\theta',\theta'')\,A(\theta,\theta'')
},
\end{equation}
where
\begin{equation}
\mathcal Q_{\rm B}^{(\theta)}(\theta',\theta'')
=
\tr\!\left[
(\partial_{\theta'}\rho[\theta'])\,
\Omega_{\rho[\theta]}\!\bigl(\partial_{\theta''}\rho[\theta'']\bigr)
\right].
\end{equation}
Here $\Omega_{\rho[\theta]}$ is the generalized SLD superoperator introduced in Sec.~\ref{SM_cl_to_quantum}. Equation~\eqref{eq:SM_bayesian_bound_qm} is the quantum counterpart of Eq.~\eqref{eq:SM_bayesian_bound} and reduces to the quantum Van Trees bound in the local limit $\epsilon\to 0$.

\subsection{Example}

We again consider the noisy qubit state
\begin{equation}
\rho(\theta)=\frac{1+\vect r_\theta^\top\vect\sigma}{2},
\qquad
\vect r_\theta=(d\cos\theta,d\sin\theta,0),
\end{equation}
with Gaussian prior
\begin{equation}
p(\theta)=\frac{1}{\sqrt{2\pi}\sigma_p}\exp\!\left(-\frac{\theta^2}{2\sigma_p^2}\right).
\end{equation}
As an explicit example, we choose a Gaussian hierarchy function
\begin{equation}
A_\epsilon(\theta,\theta')
=
\frac{1}{\sqrt{2\pi}\epsilon}
\exp\!\left[-\frac{(\theta-\theta')^2}{2\epsilon^2}\right].
\end{equation}
Define
\begin{equation}
\mathbb G_A(\theta)
\equiv
\int d\theta'\,
A_\epsilon(\theta,\theta')\,\partial_{\theta'}\rho[\theta'].
\end{equation}
Then Eq.~\eqref{eq:SM_bayesian_bound_qm} becomes
\begin{equation}\label{SM_g_vtb}
\mathrm{BMSE}(\hat\theta)
\ge
\frac{1}{
\int d\theta\,
\tr\!\left[
\mathbb G_A(\theta)\,
\Omega_{\rho[\theta]}\!\bigl(\mathbb G_A(\theta)\bigr)
\right]
}.
\end{equation}

Let $\ket{\pm_\theta}$ be the eigenstates of $\rho(\theta)$, with eigenvalues
\begin{equation}
p_{\theta;\pm}=\frac{(1\pm d)p(\theta)}{2}.
\end{equation}
By the definition of the generalized SLD operator,
\begin{align}\label{SM_B_eg_1}
\tr\!\left[
\mathbb G_A(\theta)\Omega_{\rho[\theta]}\!\bigl(\mathbb G_A(\theta)\bigr)
\right]
&=
\frac{2[\mathbb G_A(\theta)]_{++}^2}{(1+d)p(\theta)}
+
\frac{2[\mathbb G_A(\theta)]_{--}^2}{(1-d)p(\theta)}
+
\frac{4|[\mathbb G_A(\theta)]_{+-}|^2}{p(\theta)},
\end{align}
where
\begin{equation}
[\mathbb G_A(\theta)]_{\alpha\beta}
=
\bra{\alpha_\theta}\mathbb G_A(\theta)\ket{\beta_\theta},
\qquad
\alpha,\beta=\pm.
\end{equation}

A direct calculation gives
\begin{align}\label{SM_b_eg_2}
[\mathbb G_A(\theta)]_{++}
&=
\int d\theta'\,
A_\epsilon(\theta,\theta')
\left[
\frac{\partial_{\theta'}p(\theta')}{2}\bigl(1+d\cos(\theta-\theta')\bigr)
+
\frac{dp(\theta')}{2}\sin(\theta-\theta')
\right],
\nonumber\\
[\mathbb G_A(\theta)]_{--}
&=
\int d\theta'\,
A_\epsilon(\theta,\theta')
\left[
\frac{\partial_{\theta'}p(\theta')}{2}\bigl(1-d\cos(\theta-\theta')\bigr)
-
\frac{dp(\theta')}{2}\sin(\theta-\theta')
\right],
\nonumber\\
[\mathbb G_A(\theta)]_{+-}
&=
\frac{id}{2}\int d\theta'\,
A_\epsilon(\theta,\theta')
\left[
-\partial_{\theta'}p(\theta')\sin(\theta-\theta')
+
p(\theta')\cos(\theta-\theta')
\right].
\end{align}

To evaluate these expressions, define
\begin{align}
f_1(\theta)
&=
\int d\theta'\,
A_\epsilon(\theta,\theta')\,p(\theta')\,e^{i\theta'},
\nonumber\\
f_2(\theta)
&=
\int d\theta'\,
A_\epsilon(\theta,\theta')\,\partial_{\theta'}p(\theta')\,e^{i\theta'},
\nonumber\\
f_3(\theta)
&=
\int d\theta'\,
A_\epsilon(\theta,\theta')\,\partial_{\theta'}p(\theta').
\end{align}
For the Gaussian prior and Gaussian hierarchy kernel, these integrals can be evaluated analytically. Introducing
\begin{equation}
\Sigma^2\equiv \epsilon^2+\sigma_p^2,
\end{equation}
we obtain
\begin{align}
f_1(\theta)
&=
\frac{1}{\sqrt{2\pi\Sigma^2}}
\exp\!\left[
-\frac{\theta^2}{2\Sigma^2}
-\frac{\epsilon^2\sigma_p^2}{2\Sigma^2}
+i\frac{\sigma_p^2\theta}{\Sigma^2}
\right],
\nonumber\\
f_3(\theta)
&=
-\frac{\theta}{\Sigma^2\sqrt{2\pi\Sigma^2}}
\exp\!\left(-\frac{\theta^2}{2\Sigma^2}\right).
\end{align}
Using
\begin{equation}
\frac{d f_1(\theta)}{d\theta}
=
-\frac{\theta}{\epsilon^2}f_1(\theta)
-\frac{\sigma_p^2}{\epsilon^2}f_2(\theta),
\end{equation}
we further find
\begin{equation}
f_2(\theta)
=
-f_1(\theta)\frac{\theta+i\epsilon^2}{\Sigma^2}.
\end{equation}

It is convenient to define
\begin{equation}
g(\theta)
=
\frac{1}{\sqrt{2\pi\Sigma^2}}
\exp\!\left[
-\frac{\theta^2}{2\Sigma^2}
-\frac{\epsilon^2\sigma_p^2}{2\Sigma^2}
\right],
\qquad
\alpha=\frac{\epsilon^2}{\Sigma^2},
\end{equation}
so that
\begin{equation}
[f_1(\theta)]^*e^{i\theta}=g(\theta)e^{i\alpha\theta},
\qquad
[f_2(\theta)]^*e^{i\theta}
=
-g(\theta)\frac{\theta-i\epsilon^2}{\Sigma^2}e^{i\alpha\theta}.
\end{equation}
Substituting these expressions into Eq.~\eqref{SM_b_eg_2} gives
\begin{align}\label{SM_b_eg_3}
[\mathbb G_A(\theta)]_{++}
&=
\frac{f_3(\theta)}{2}
-\frac{d\,g(\theta)}{2\Sigma^2}
\left[\theta\cos(\alpha\theta)+\epsilon^2\sin(\alpha\theta)\right]
+\frac{d\,g(\theta)}{2}\sin(\alpha\theta),
\nonumber\\
[\mathbb G_A(\theta)]_{--}
&=
\frac{f_3(\theta)}{2}
+\frac{d\,g(\theta)}{2\Sigma^2}
\left[\theta\cos(\alpha\theta)+\epsilon^2\sin(\alpha\theta)\right]
-\frac{d\,g(\theta)}{2}\sin(\alpha\theta),
\nonumber\\
[\mathbb G_A(\theta)]_{+-}
&=
i\frac{d\,g(\theta)}{2\Sigma^2}
\left[\theta\sin(\alpha\theta)+\sigma_p^2\cos(\alpha\theta)\right].
\end{align}
Substituting Eq.~\eqref{SM_b_eg_3} into Eq.~\eqref{SM_B_eg_1} yields the generalized Bayesian bound~\eqref{SM_g_vtb} as a function of $\epsilon$, which is used to generate Fig.~2 in the main text.



\bibliography{02ref}

\end{document}